\documentclass[a4paper,UKenglish,cleveref, autoref, thm-restate]{lipics-v2021}
\usepackage{graphicx} 
\usepackage{lineno}
\usepackage{url}
\usepackage{float}
\usepackage{amsmath,amsthm,latexsym}
\usepackage{xcolor}
\usepackage{multicol}
\usepackage{bm}
\usepackage{xspace}
\usepackage{url}
\usepackage{hyperref}
\usepackage[capitalise]{cleveref}
\usepackage{tikzducks}
\usepackage{graphicx}
\usepackage[most]{tcolorbox}
\usepackage{adjustbox}
\usepackage{tikz}
\usepackage[colorinlistoftodos,prependcaption,textsize=tiny]{todonotes} 
\usepackage{array}
\usepackage{multirow}
\usepackage{thmtools}
\usepackage{mathtools}
\usepackage{thm-restate}

\usetikzlibrary{decorations.pathreplacing}
\usetikzlibrary{automata}
\tikzset{>=stealth, shorten >=1pt}
\tikzset{every edge/.style = {thick, ->, draw}}
\tikzset{every loop/.style = {thick, ->, draw}}
\usetikzlibrary{calc}
\usetikzlibrary{arrows,automata,decorations.pathmorphing}
\usetikzlibrary{positioning}
\usetikzlibrary{shapes,shapes.geometric,fit,calc,automata,}

\newcommand{\ptime}{{\sc PTime}\xspace}
\newcommand{\pspace}{{\sc PSpace}\xspace}
\newcommand{\exptime}{{\sc ExpTime}\xspace}

\newcommand{\Subject}[1]{\paragraph*{#1.}}
\newcommand{\St}{~|~}

\newcommand{\tuple}[1]{\langle #1  \rangle}
\newcommand{\pair}{\tuple}

\newcommand{\Nat}{\ensuremath{\mathbb{N}}\xspace}
\newcommand{\PNat}{\Nat^+\xspace}

\newcommand{\HD}{\ensuremath{\mathrm{HD}}\xspace}

\newcommand{\A}{{\mathcal{A}}}
\newcommand{\B}{{\mathcal{B}}}
\newcommand{\C}{{\mathcal{C}}}
\newcommand{\D}{{\mathcal{D}}}

\newcommand{\G}{{\mathcal{G}}}
\newcommand{\M}{{\mathcal{M}}}                      
\newcommand{\K}{{\mathcal{K}}}
\newcommand{\T}{{\mathcal{T}}}
\newcommand{\N}{{\mathcal{N}}}
\renewcommand{\P}{{\mathcal{P}}}
\renewcommand{\S}{{\mathcal{S}}}
\newcommand{\R}{{\mathcal{R}}}

\newcommand{\Ghost}[1]{\textsf{#1-TokenGhost}}

\newcommand{\restr}[2]{\left.#1\right|_{#2}}

\newcommand{\content}{\text{content}\xspace}

\newcommand{\trans}[3]{#1\xrightarrow[]{#2}#3}

\newcommand{\Func}[1]{{\mathsf{#1}}}
\newcommand{\Delay}{\Func{Delay}}
\newcommand{\Frac}{\Func{Frac}}
\newcommand{\SigmaShift}{\Func{Delay}}
\newcommand{\reject}{\Func{Reject}}

\newcommand{\LimInf}{\Func{LimInf}}
\newcommand{\LimSup}{\Func{LimSup}}

\newcommand{\DSum}{\Func{DSum}}

\newcommand{\Sim}{Sim}

\renewcommand{\phi}{\varphi}

\newcommand{\eve}{\pmb{E}}

\newcommand{\op}{\textsf{op}}
\newcommand{\push}{\textsf{push}}
\newcommand{\pop}{\textsf{pop}}
\newcommand{\noop}{\textsf{noop}}

\definecolor{Green}{rgb}{0.13, 0.55, 0.13}

\newcommand{\strat}{\mathsf{s}}

\renewcommand{\epsilon}{\varepsilon}

\crefname{enumi}{}{}

\newcommand{\TableRef}[1]{{\small \cref{#1}}}

\linenumbers
\title{History-Determinism vs Fair Simulation}
\titlerunning{History-Determinism vs Fair Simulation}
\author{Udi Boker}{Reichman University, Herzliya, Israel \and \url{https://faculty.runi.ac.il/udiboker/}}{udiboker@runi.ac.il}{0000-0003-4322-8892}{Israel Science Foundation grant 2410/22}
\author{Thomas A.  Henzinger}{Institute of Science and Technology Austria (ISTA), Klosterneuburg, Austria \and \url{https://pub.ista.ac.at/~tah/}}{tah@ista.ac.at}{0000-0002-2985-7724}{}
\author{Karoliina Lehtinen}{CNRS, LIS, Aix-Marseille Univ. \and \url{https://lehtinenkaroliina.wordpress.com/}}{lehtinen@lis-lab.fr}{0000-0003-1171-8790}{ANR QUASY 23-CE48-0008-01}
\author{Aditya Prakash}{University of Warwick, Coventry, UK \and \url{https://apitya.github.io/}}{aditya.prakash@warwick.ac.uk}{https://orcid.org/0000-0002-2404-0707}{Chancellors' International Scholarship from the University of Warwick and Centre for Discrete Mathematics and Its Applications (DIMAP)}

\authorrunning{U. Boker, T. A. Henzinger, K. Lehtinen, and A. Prakash} 

\Copyright{U. Boker, T. A. Henzinger, K. Lehtinen, and A. Prakash} 



\category{} 

\relatedversion{} 

\nolinenumbers 

\EventEditors{2}
\EventNoEds{2}
\EventLongTitle{Conference of Very Important Things}
\EventShortTitle{CVIT 2016}
\EventAcronym{CVIT}
\EventYear{2016}
\EventDate{December 24--27, 2016}
\EventLocation{Little Whinging, United Kingdom}
\EventLogo{}
\SeriesVolume{42}
\ArticleNo{23}

\begin{document}

\maketitle

\begin{abstract}
An automaton $\A$  is history-deterministic if its nondeterminism can be resolved on the fly, only using the prefix of the word read so far. This mild form of nondeterminism has attracted particular attention for its applications in synthesis problems. An automaton $\A$ is guidable with respect to a class $C$ of automata if it can fairly simulate every automaton $\B$ in $C$, whose language is contained in that of $\A$. In other words, guidable automata are those for which inclusion and simulation coincide, making them particularly interesting for model-checking. 

We study the connection between these two notions, and specifically the question of when they coincide. For classes of automata on which they do, deciding guidability, an otherwise challenging decision problem, reduces to deciding history-determinism, a problem that is starting to be well-understood for many classes.

We provide a selection of sufficient criteria for a class $C$ of automata to guarantee the coincidence of the notions, and use them to show that the notions coincide for the most common automata classes, among which are $\omega$-regular automata and many infinite-state automata with safety and reachability acceptance conditions, including vector addition systems with states, one-counter nets, pushdown-, Parikh-, 
and timed-automata.

We also demonstrate that history-determinism and guidability do not always coincide, for example, for the classes of timed automata with a fixed number of clocks.

\end{abstract}





\section{Introduction}
Language inclusion between automata is a key problem in verification: given an automaton representing a program and another representing a specification, language inclusion of the former in the latter captures precisely whether all executions of the program satisfy the specification. Unfortunately, in the presence of nondeterminism, inclusion is algorithmically hard. For instance, for regular automata it is \pspace-hard on both finite and infinite words.

Fair simulation is a more syntactic approximation of inclusion, defined by the simulation game~\cite{HKR02}. In this game, one player, in the role of the spoiler, builds, transition by transition, a run in one of the automata, say $\A$, while the other, in the role of the duplicator,  chooses at each turn a matching transition in the other automaton, say $\B$. The second player's task is to build a run that is accepting if the first player's run is accepting. If the duplicator has a winning strategy, then $\B$ is said to simulate $\A$, which, in particular, implies that $\A$'s language is included in $\B$'s language. Simulation, due to its local and syntactic nature, is generally easier to check than inclusion: for instance it is in \ptime~ for nondeterministic B\"uchi automata. As a result, automata for which language inclusion and simulation coincide are particularly well-suited for model-checking. We call such automata \textit{guidable}, after a similar notion used previously as an alternative to determinism for tree automata~\cite{CL08}. Despite their clear usefulness for model checking, guidable automata have so far been mostly used as a tool, but not studied much in their own right, with the notable exception of~\cite{NS21}.

Guidability is not easy to decide: it is contingent on an automaton simulating a potentially infinite number of language-included automata. We would like to have, whenever possible, a characterisation that is more amenable to algorithmic detection. 

Deterministic automata are of course always guidable, and so are \textit{history-deterministic} automata. These are mildly nondeterministic automata, in which nondeterministic choices are permitted, but can only depend on the word read so far, rather than the future of the word. These automata have received a fair bit of attention recently due, in particular, to their applications in synthesis problems~\cite{BL23}. In general, they offer an interesting compromise between the power of nondeterministic automata and the better algorithmic properties of deterministic ones. In particular, they can simulate all equivalent, or language-contained, automata as they only need the history to resolve nondeterministic choice in the best possible way---in other words, they are guidable. 
 In fact, at first it might appear that history-determinism and guidability should coincide; indeed, this is the case if we consider guidability with respect to all labelled transition systems \cite[Theorem 4]{HLT22}. However, there are also classes of automata for which this is not the case.
 

Guidability and history-determinism coinciding on a class $C$ of automata is equivalent to the description \textit{``for every automaton $\A\in C$ that is not history-deterministic, there is some $\A'\in C$ that is language-included in $\A$ but that $\A$ does not simulate'' (D).} Then it is easy enough to hand-pick automata to build classes where guidability and history-determinism do not coincide (for example, a class of inclusion-incomparable automata that are not all history-deterministic). However, as we will see, there are also more natural classes of automata, such as timed automata with a bounded number of clocks, for which guidability and history-determinism differ.

The characterisation (D) is  too abstract to be much use for analysing the usual automata classes we are interested in. We therefore prove that several more concrete criteria  (\cref{cl:Criteria}) imply that guidability and history-determinism coincide, and use each of these in a comprehensive analysis of standard automata classes.
Roughly, each of the criteria describes some sufficient closure properties which guarantee the existence of $\A'$ from description (D). If some automaton can simulate another automaton that is sufficiently difficult to simulate (for example a deterministic one, since they simulate all equivalent automaton), then it must be history-deterministic; as a result, a class of automata having sufficient closure properties (such as closure under determinisation), implies that guidability and history-determinism coincide. The challenge is to identify, for a variety of different classes of automata $C$, an automaton that is sufficiently difficult to simulate, while remaining in $C$.

In order to discuss our sufficient criteria in more detail, we need to start with a couple of key notions, the first of which is the \textit{1-token game} and the second is the \textit{1-token ghost}.

History-determinism of an automaton is tricky and expensive to decide directly~\cite{HP06}; As an alternative, Bagnoal and Kuperberg used $k$-token games as potential characterisations for history-determinism~\cite{BK18}. Roughly, they resemble a (fair) simulation-like game, played on a single automaton, where one player, Eve (in the role of Duplicator), must build transition-by-transition a run on a word dictated letter-by-letter by Adam, who, after each of Eve's choices, also builds $k$ runs transition-by-transition. Eve then wins if her run is accepting whenever Adam's run is accepting. Bagnol and Kuperberg showed that for B\"uchi automata, the $2$-token game indeed characterises history-determinism, which means that deciding history-determinism for B\"uchi automata is in \ptime~\cite{BK18}. Since then, the $1$- or $2$-token games have been shown to characterise history-determinism for various automata classes, including coB\"uchi~\cite{BKLS20b}, $\DSum$, $\LimInf$, and $\LimSup$ automata~\cite{BL22}.
These games contrast with the \textit{letter game}, a game which always characterises history-determinism, but which is often challenging to solve directly~\cite{HP06}.

In this work, we make heavy use of token-games, this time to understand the connection between history-determinism and guidability. 
In particular, we extend token games to be played over two automata (\cref{def:TokenGamesTwoAutomata}), separating between Eve's and Adam's automata, and define, given an automaton $\A$, that an automaton $\A'$ is a  \textit{$1$-token ghost} of $\A$ if $\A'$ is language-equivalent to $\A$ and Eve wins the $1$-token game between $\A'$ and $\A$.
For some classes of automata such a ghost is easy to construct, while for other classes it might not exist due to lacking closure properties. 

With these notions, we can now state our criteria.

\begin{theorem}\label{cl:Criteria}
The notions of history-determinism and guidability coincide for a class $C$ of labelled transitions systems (LTSs) if at least one of the following holds:
\begin{enumerate}
     \item \label{Item:Determinisation} 
     \emph{Determinisation.} $C$ is closed under history-determinism, i.e., for each nondeterministic LTS in $C$, there is a language-equivalent history-deterministic (or deterministic) LTS in $C$ as well.
     (\cref{cl:WhenNdet=Hdet})
    \item  \label{Item:OneTokenGhost}
    \emph{$1$-token ghost.} $1$-token games characterise history-determinism in $C$, and $C$ is closed under $1$-token ghost. (\cref{cl:TokenCriteria})
    \item  \label{Item:SafetyProduct}
    \emph{Strategy ghost.}  For every $\A\in C$ that is not history-deterministic, there is a deterministic LTS $\B$ over the alphabet of transitions of $\A$, such that $\B$ recognises the plays of a winning strategy of Adam in the letter game on $\A$, and $\B$, projected onto the alphabet of $\A$, has a $1$-token ghost in $C$.  (\cref{cl:CompositionCriterion})
\end{enumerate}
\end{theorem}

We prove the criteria of \cref{cl:Criteria} in \cref{sec:criteria}, and use them in \cref{sec:Yes} to show that the notions of history-determinism and guidability coincide for numerous classes of automata, listed in~\cref{table:table1}.
We also provide counter-examples of classes for which history-determinism and guidability do not coincide, as listed in ~\cref{table:table1}, and elaborated on in \cref{sec:No}.
We restrict our analysis to automata over infinite words, which are better behaved in this context, and discuss how to adapt our techniques for finite words in~\cref{sec:conclusions}.

\newcolumntype{C}[1]{>{\centering\let\newline\\\arraybackslash\hspace{0pt}}m{#1}}

\begin{table*}[ht!]
    \centering
    \def\arraystretch{1.2} 
    
    \begin{tabular}    { | C{7cm} | C{4.4cm}| } 
    \hline
    Automata Class &  HD\,=\,Guidability by  \\ 
    \hline\hline
    
%
   $\omega$-regular &  \emph{Determinisation\newline or Strategy ghost}\\
	\hline
   Fixed-index parity (e.g., B\"uchi), Weak & \emph{Strategy ghost}
   \TableRef{cl: hd-guid parity} 
	\\ \hline
	
   Linear & \emph{Strategy ghost}  ~\TableRef{cl:LinearHdGuidability}
   	.  (Not ghost-closed \TableRef{cl:LinearNoDelay})  
    \\	\hline
    
	 Safety \& Reachability\newline pushdown automata, VASS, timed, one-counter automata, one-counter net, Parikh &  \emph{$1$-token ghost}
	\newline \TableRef{cl:timedhd-guid,cl:hd-guid-pushdownparikhvass}
	\\   \hline
	
	 VPA with any $\omega$-regular acceptance condition &  \emph{Strategy ghost} \TableRef{cl:hd-guid-vpa} 
	\\   \hline

	

	
	\multicolumn{2}{|l|}{Classes for which \HD $\neq$ guidability:}\\
	\multicolumn{2}{|l|}{- B\"uchi automata with a bounded number of states.~~\TableRef{cl:BoundedStates}}  \\
	\multicolumn{2}{|l|}{- Timed automata with a bounded number of clocks.~~\TableRef{cl:1-clock}}  \\

	\hline \hline
	\multicolumn{2}{|l|}{Notable classes for which we leave the question \HD $=$ \hspace{-12pt} ? guidability open:}\\
	\multicolumn{2}{|l|}{- PDA/OCA/OCN/Timed automata with general $\omega$-regular acceptance}  \\
	\hline
	
\end{tabular}
    \caption{History-determinism vs guidability
    \label{table:table1}}
\end{table*}

A practical corollary of our result is that guidability is decidable, with the complexity of deciding history-determinism, for $\omega$-regular automata (\exptime~for parity automata, \ptime for B\"uchi and coB\"uchi), safety and reachability timed automata (\exptime) and visibly pushdown automata (\exptime). For details on the complexity of these procedures, we refer the reader to a recent survey~\cite{BL23}.

\section{Preliminaries}
We use $\Nat$ and $\PNat$ to denote the set of non-negative and positive integers respectively. We use $[i..j]$ to denote the set $\{i,\ldots,j\}$ of integers, and $[i]$ for the set $[1..i]$.
An alphabet $\Sigma$ is a non-empty set of letters. A finite or infinite word is a finite or infinite sequence of letters from $\Sigma$ respectively. We let $\epsilon$ be the empty word, and $\Sigma_{\epsilon}$ the set $\Sigma\cup \{\epsilon\}$. The set of all finite (resp. infinite) words is denoted by $\Sigma^{*}$ (resp. $\Sigma^{\omega}$). A \emph{language} is a set of words. 

\subsection{Labelled transition systems} 
A \emph{labelled transition system} (LTS) $\A=(\Sigma,Q, \iota, \Delta,\alpha)$ consists of a potentially infinite alphabet $\Sigma$, a potentially infinite state-space $Q$, an initial state $\iota\in Q$, a labelled transition relation $\Delta\subseteq Q\times\Sigma_{\epsilon}\times Q$, and a set of accepting runs $\alpha$, where a run $\rho$ is a (finite or infinite) sequence of transitions starting in $\iota$ and following $\Delta$.
We may write $\trans{q}{\sigma}{q'}$ instead of $(q,\sigma,q')\in\Delta$ for $\sigma \in \Sigma_{\epsilon}$. Given a finite run $\rho = q_0 \xrightarrow{\sigma_1} q_1 \xrightarrow{\sigma_2} \cdots \xrightarrow{\sigma_k} q_k$ on the word $v = \sigma_0 \cdot \sigma_1 \cdot \sigma_2 \cdots \sigma_k \in \Sigma_{\epsilon}^{k+1}$, we write $q_0 \xRightarrow{v,\rho} q_k$ to denote that $\rho$ is a transition sequence on the word $v$ that starts at $q_0$ and ends at $q_k$.
An LTS is \emph{deterministic} if for every state $q$ and letter $\sigma \in \Sigma$, there is at most one transition $q \xrightarrow{\sigma} q'$ from $q$ on the $\sigma$, and there are no transitions on $\epsilon$.


A word $w\in\Sigma^{\omega}$ is accepted by an LTS $\A$ if there is an accepting run of $\A$ on $w$. The language $L(\A)$ of $\A$ is the set of words that it accepts. An LTS $\A$ is contained in an LTS $\B$, denoted by $\A\leq\B$, if $L(\A) \subseteq L(\B)$, while $\A$ and $\B$ are equivalent, denoted by $\A \equiv \B$ if $L(\A) = L(\B)$. 
 %
 
A \emph{transducer} is like an LTS without acceptance condition and $\epsilon$-transitions but with an output labelling: it is given by a tuple $(\Sigma_I, \Sigma_O, Q,\iota, \Delta, \gamma)$, where $\Sigma_I$ and $\Sigma_O$ are the input and output alphabets, respectively, $\Delta\subseteq Q\times\Sigma_I\times Q$ is the transition relation, and $\gamma:\Delta\to\Sigma_O$ is the output function. 
A \emph{strategy} in general is a deterministic transducer. It is finite memory if the transducer has a finite state space.

\subsection{Automata} 
We briefly recall the automata types considered here, which we assume to operate on infinite words unless stated otherwise, and leave the formal definitions to the appendix.

LTSs are represented concisely by various automata. An automaton $\A$ induces an LTS $\B$, whose states are the configurations of $\A$, and whose runs are the same as $\A$'s runs. If $\A$'s states and configurations are the same, as is the case with $\omega$-regular automata which we define below, then $\B$ is identical to $\A$, but with the acceptance condition given by a set of runs (as opposed to an $\omega$-regular conditions). If the configurations of $\A$ contain additional data, as is the case for example with pushdown automata, then $\B$ and $\A$ have different states. Notice that $\A$ is deterministic iff $\B$ is.  
The acceptance condition of $\A$ induces the acceptance condition on $\B$. 

 In a parity condition, $\alpha$ assigns priorities in $\mathbb{N}$ to either states or transitions, and a run is accepting if the highest priority that occurs infinitely often along it is even. An $[i,j]$-parity automaton, for $i<j$ two natural numbers is a parity automaton whose priorities are in $[i,j]$. A parity automaton is said to be a \emph{weak automaton} if there is no cycle in the automaton containing both an even and odd priority. In a reachability condition, some states are labelled final; a run accepts if it reaches a final state. In a safety automaton, some states are labelled safe; a run accepts if it remains within the safe region. 

In a nondeterministic $\omega$-regular automaton $(\Sigma,Q,\iota,\delta,\alpha)$, $\Sigma$ and $Q$ are finite, and the acceptance condition $\alpha$ is based on the set of states (or transitions) visited infinitely often along a run.
A timed automaton (TWA) $(\Sigma, Q, \iota, C,\delta,\alpha)$ has a set of clocks $C$ and its transitions are guarded by inequalities between the clock values and can reset clocks. 
It reads timed words, which consist of letters of  a finite alphabet $\Sigma$ paired with delays from $\mathbb{R}$. A timed automaton recognises a timed language, for example ``at some point an event occurs twice exactly one time-unit apart.''

We also handle pushdown automata, one-counter automata, vector addition systems with states, one-counter nets and Parikh automata with reachability and safety acceptance. We define these classes of infinite state systems uniformly in~\cref{sec:infinite-state} as classes of finite state automata with transitions that modify an infinite \textit{content space}. A visibly pushdown automaton (VPA) is a pushdown automaton without $\varepsilon$-transitions, in which the input alphabet is partitioned into $\pop,\push$ and $\noop$ letters that induce only transitions that $\pop,\push$ and have no effect on the stack respectively.

\section{History-determinism, simulation and related games}\label{sec:games}
Different simulation-like games that capture either a relationship between automata or properties of a single automaton are at the heart of our technical developments. In this section we go over the various games---both known and newly defined---that will be played throughout this article, and which allow us to connect guidability and history-determinism.
These games are all based on Adam (the spoiler) building a word letter by letter, and potentially a run in an automaton over that word, while his opponent Eve (the duplicator) tries to build a single accepting run transition by transition. The differences between these games are based on whether they are played on one or two automata, whether Adam picks transitions, and if so, whether he does it before Eve. The winning condition is similar in all cases: Eve's run must be accepting whenever Adam's run is accepting, or if Adam does not have a run, then whenever the word built by Adam's moves is in the language of a specified automaton. This results in three styles of games: (i) simulation games, played on two automata, in which Adam plays before Eve, and each builds a run in their respective automaton; (ii) token-games, which can be played on one or two automata,  in which Adam first declares the letter, and then Eve plays her transition before Adam plays his; and (iii) the letter game, played on a single automaton, in which Adam only chooses letters and does not build a run at all.\\

Fair simulation between two  LTSs (or automata) is captured by the simulation game defined below:
\begin{definition}[Simulation game]
	Consider LTSs $\A = (\Sigma, Q_A, \iota_A,\Delta_A, \alpha_A)$ and $\B = (\Sigma, Q, \iota_B, \Delta_B, \alpha_B)$. 
	The simulation game $\Sim(\B,\A)$ between $\B$ and $\A$ is a two player-game played between Adam and Eve with positions in $Q_A \times Q_B$ which starts at the position $(p_0,q_0) = (\iota_A,\iota_B)$. At round $i$ of the play, for $i \geq 0$, when the position is $(p_i,q_i)$:
	\begin{itemize}
		\item Adam picks  $\sigma_i \in \Sigma$ and a transition (or transition sequence, in the presence of $\epsilon$-transitions) $p_i \xRightarrow{\sigma_i,\rho_i} p_{i+1}$ in $\A$;
		\item Eve picks a transition (or transition sequence, in the presence of $\epsilon$-transitions)  $q_i \xRightarrow{\sigma_i,\rho'_i} q_{i+1}$ in $\B$; they proceed from $(p_{i+1},q_{i+1})$.
	\end{itemize}
	An infinite play produces a run $\rho_A$ in $\A$ consisting of transitions chosen by Adam and a run $\rho_E$ in $\B$ of transitions chosen by Eve, both on $\sigma_0 \sigma_1 \sigma_2 \cdots $. We say that Eve wins the play if $\rho_E$ is accepting or $\rho_A$ is rejecting.   
\end{definition}
If Eve has a winning strategy in this game, we say that $\B$ simulates $\A$, denoted by $\Sim(\B,\A)$. It is easy to observe that if $\B$ simulates $\A$, then $L(\A)\subseteq L(\B)$.
An LTS $\A$ is \emph{guidable} with respect to a class $C$ of LTSs if $\A$ simulates every LTS $\A'$ in $C$ that satisfies $L(\A')\subseteq L(\A)$.

The following \emph{letter game} based definition of history-determinism, was introduced by Henzinger and Piterman~\cite{HP06}, and coincides with Colcombet's notion of translation strategies~\cite{Col13}.

Given an LTS $\A$, the \textit{letter game} on $\A$, denoted by $\HD(\A)$ is similar to the simulation game except that instead of playing transitions in an automaton, Adam just chooses letters and builds a word $w$, letter by letter, which should be in the language of $\A$, while Eve tries to build a run of $\A$ over $w$. More precisely, the letter game starts with Eve's token at the initial state $\iota$, and proceeds in rounds. At round $i$, where Eve's token is at $q_i$:
\begin{itemize}
	\item Adam chooses a letter $\sigma_i$ in the alphabet $\Sigma$ of $\A$;
	\item Eve chooses a transition  $\trans{q_i }{\sigma_i}{q_{i+1}}$ (or a transition sequence $q_i \xRightarrow{\sigma_i,\rho_i}q_{i+1}$ in the presence of $\epsilon$-transitions) of $\A$ over $\sigma_i$; Eve's token moves to $q_{i+1}$.
\end{itemize}
In the limit, a play consists of the word $w=\sigma_0\sigma_1\cdots$ and the run $\rho=\rho_0 \cdot \rho_1 \cdot \rho_2 \cdots$. Eve wins the play if $w\notin L(\A)$ or $\rho$ is accepting. We say that $\A$ is history-deterministic (\HD) if Eve has a winning strategy in the letter game over $\A$.

Token games are known to characterise history-determinism on various classes of automata~\cite{BK18,BL22,BL23}. We generalise token games to be played on two LTSs below, which makes them more akin to a variation of simulation. This will help us relate simulation and history-determinism in~\cref{sec:criteria}. We only use the $1$-token version here.

\begin{definition}[$1$-token games over two LTSs (or automata)]\label{def:TokenGamesTwoAutomata}
	Consider LTSs $\A'$ and $\A$ with initial states $p_0$ and $q_0$ respectively. In the $1$-token game on $\A'$ and $\A$ denoted by $G_1(\A', \A)$, Eve has a token with which she constructs a run in $\A'$, and Adam has a token with which he constructs a run in $\A$. The game proceeds in rounds, and at round $i$ of the play with token positions $(p_i,q_i)$, for each $i \geq 0$:
	\begin{itemize}
		\item Adam chooses a letter $\sigma_i$ in $\Sigma$;
		\item Eve chooses a transition (or a transition sequence, in the presence of $\varepsilon$-transitions) $p_i \xRightarrow{\sigma_i,\rho'_i} p_{i+1}$ in $\A'$;
		\item Adam chooses a transition (or transition sequence, in the presence of $\varepsilon$-transitions) $q_{i} \xRightarrow{\sigma_i,\rho_i} q_{i+1}$; the game proceeds from $(p_{i+1},q_{i+1})$.
	\end{itemize} 
 An infinite play produces a word $w=\sigma_0\dots$, a sequence of transitions $\rho_E$ of $\A'$ chosen by Eve, and a sequence of transitions $\rho_A$ in $\A$ chosen by Adam. Eve wins if $\rho_E$ is accepting or if $\rho_A$ is rejecting.
\end{definition}

A strategy for Eve here is a function $\strat:(\Delta^{+})^* \times \Sigma \to (\Delta')^{*}$, where $\Sigma$ is the alphabet of $\A$ and $\A'$, and $\Delta$ and $\Delta'$ are the sets of transitions of $\A$ and $\A'$, respectively.
When clear from context,  $G_1(\A', \A)$ also denotes the claim that Eve has a winning strategy in the game $G_1(\A', \A)$. 
As an automaton and its induced LTS have the same runs, $G_1(\A',\A)$ holds for automata $\A$ and $\A'$ iff it holds for their induced LTSs. We also write $G_1(\A)$ for $G_1(\A,\A)$.

Note that the $1$-token game and the simulation game differ in one key aspect: in the simulation game, Adam plays first, and Eve can use the information of the transition to inform her choice, while in the $1$-token game, Eve must choose her transition based only on the \textit{letter} chosen by Adam, who plays his transition after Eve. 






\section{Criteria for when History-Determinism = Guidability}\label{sec:criteria}
We now provide criteria which guarantee that history-determinism and guidability coincide for a class of LTSs. In \cref{sec:Yes}, we use these to show the coincidence of the two notions for many standard automata classes.

\subsection{Closure under (history-)determinism}\label{sec:criterion-HD.tex}

A first observation is that if every LTS $\A$ can be determinised within the class $C$, or even if there exists an equivalent \HD LTS $\A'$ within $C$ then $\A$ is \HD if and only if it is guidable.

\begin{restatable}{lemma}{clwhenndethdet}\label{cl:WhenNdet=Hdet}
	History-determinism and guidability coincide for any class $C$ of LTSs in which the languages expressed by history-deterministic (or deterministic) LTSs are the same as languages expressed by nondeterministic LTSs.
\end{restatable}

The proof is simple: one direction is trivial (\HD always implies guidability) and conversely, if an automaton $\A$ is not \HD, then it cannot simulate any equivalent \HD automaton, implying that $\A$ is not guidable.

Various examples of such classes are provided in \cref{sec:Straightforward}, as summarised in~\cref{table:table1}.
In particular, the general class of all labelled-transition systems~\cite{HLT22}, safety/reachability visibly pushdown automata~\cite{AM04}, as well as finite-state automata on finite words (NFAs), and co-B\"uchi, Parity, Rabin, Streett, and Muller automata on infinite words. Yet, this is not the case for B\"uchi automata or parity automata with a fixed parity index. History-determinism is also strictly less expressive than nondeterminism for pushdown automata, Parikh automata, timed automata and one-counter nets.

\subsection{Via token games}\label{sec:criterion-tokens}

For classes that are not closed under determinisation, we have to find some other type of automaton that is difficult to simulate. To do so, we revisit token games, previously used to help decide history-determinism, to relate history-determinism and guidability.
Recall that we extended the definition of $1$-token games, so that they are played on two automata, rather than one. 
 In the next definition, we use this extended notion of $1$-token game to identify, for each automaton $\A$, an automaton $\A'$ such that Eve wins the the $1$-token game on $\A$ if and only if $\A$ simulates $\A'$.

\begin{definition}[1-token ghost]\label{def:Delay}
	An LTS (or an automaton) $\A'$ is a \emph{$1$-token ghost} of an LTS $\A$, denoted by $\Ghost{1}(\A',\A)$, if $\A'\equiv\A$ and $G_1(\A',\A)$.
\end{definition}

To show that the ghost automaton has the property that $\Sim(\A,\A')$ if and only if Eve wins $G_1(\A,\A)$, we compose the strategies in $\Sim(\A,\A')$ and $G_1(\A',\A)$. 

\begin{restatable}{lemma}{cldelaysimimpliesgk}\label{cl:Delay-ssimulation-implies-gk}
	Consider LTSs $\A$ and $\A'$, such that $\A$ simulates $\A'$ and $\Ghost{1}(\A',\A)$. Then Eve wins $G_1(\A)$.
\end{restatable}

\begin{proof}
	Let $s_{sim}$ be a winning strategy of Eve in the simulation game between $\A$ and $\A'$, and $s'$	her winning strategy in $G_1(\A',\A)$.	Eve then has a winning strategy $s$ in $G_1(\A)$: she plays the strategy $s_{sim}$ in $\Sim(\A, \A')$ against her imaginary friend Berta, who plays the strategy $s'$ in $G_1(\A',\A)$ against Adam. In more detail: In each round $i$ of the game $G_1(\A)$, Adam chooses a transition sequence $\rho_{i-1}$ in $\A$ on $\sigma_{i-1}$ (except for the first round) on his token and a letter $\sigma_i$, then Berta chooses the transition sequence $\rho^B_i = s'\left ( \rho_0 \ldots \rho_{i-1}, \sigma_i\right )$ over the letter $\sigma_i$ in $\A'$ on her token in $G_1(\A',\A)$, and then Eve chooses the transition sequence $\rho_i=s_{sim}(\rho^B_0\ldots \rho^B_{i-1})$ in $\A$. 
	
	The run built by Eve with the strategy $s$ is accepting if the run built by Berta is, which is in turn accepting if Adam's run is. Hence, $s$ is a winning strategy for Eve in $G_1(\A)$.
\end{proof}

Then, for classes in which token games characterise history-determinism and which are closed under the ghost relation, guidability and history-determinism coincide.

\begin{definition}
	A class $C$ of LTSs is \emph{closed under $1$-token ghost} if for every $\A\in C$ there exists $\A'\in C$ such that $\Ghost{1}(\A',\A)$.
\end{definition}

\begin{restatable}{lemma}{clTokenCriteria}\label{cl:TokenCriteria}
	Given a class $C$ of LTSs closed under $1$-token ghost for which $G_1$ characterises history-determinism, history-determinism and guidability coincide for $C$. 
\end{restatable}

\begin{proof}
Being \HD always implies guidability, so one direction is easy. For the other direction, if $\A$ simulates every LTS $\A'\in C$, such that $\A'\leq \A$, then in particular it simulates an LTS $\A'\in C$, such that $\Ghost{1}(\A',\A)$, as $C$ is closed under the $1$-token ghost. By \cref{cl:Delay-ssimulation-implies-gk}, Eve wins $G_1(\A)$, implying that $\A$ is \HD, as $G_1$ characterises history-determinism in $C$.
\end{proof}


A $1$-token ghost is often easy to build, by delaying nondeterministic choices by one letter (\cref{def:Simpleautomata-SigmaShift}), as shown in \cref{sec:infinite-state} for pushdown automata, one-counter automata, vector addition system with states, one-counter nets and Parikh automata.

For some automata classes, however, showing closure under $1$-token  ghosts is trickier: for VPA the stack action must occur as the letter is read, and for timed 
automata configuration updates are sensitive to the current timestamp. We handle these complications in~\cref{sec:vpa}. 
We can also only use \cref{cl:TokenCriteria} with respect to automata classes for which $1$-token games characterise history-determinism, which is not the case for parity automata or $\omega$-VPA~\cite{BK18}.
\subsection{Via Adam's strategy in the letter game}\label{sec:criterion-composition}

As we will see in detail in \cref{sec:linear}, some classes, such as linear automata, are neither closed under $1$-token ghost nor determinisation, so there is no hope for the above criteria to apply. Our final criterion is an alternative which, instead of requiring all automata to admit a $1$-ghost, builds a difficult-to-simulate automaton from Adam's winning strategy in the letter game. The intuition is that Adam's winning strategy in the letter game on an automaton $\A$ captures behaviour that is difficult for $\A$ to simulate, so if we can turn Adam's strategy into an automaton (which will be language-contained in $\A$ since Adam must play a word in the language of $\A$), then this automaton will not be simulated by $\A$. To build this automaton, we first project an automaton $\B$ recognising Adam's winning plays from his strategy onto the alphabet of $\A$, to obtain an automaton $\B_{\Sigma}$ that recognises the words played by Adam's strategy. Then, by taking the $1$-token ghost of $\B_{\Sigma}$, we obtain an automaton $\B'$ against which the simulation game is essentially the letter game against Adam's strategy. If the resulting automaton is always still in the class $C$, guidability and history-determinism coincide.


\begin{lemma}\label{cl:CompositionCriterion}
	History-determinism and guidability coincide for classes $C$ of LTSs in which, for each $\A\in C$ that is not history-deterministic, there is a deterministic LTS $\B$ over the alphabet of transitions of $\A$ that recognises the plays of a winning strategy of Adam in the letter game on $\A$, and $\B$, projected onto the alphabet of $\A$, has a $1$-token ghost in $C$. 
\end{lemma}

\begin{proof}
	Consider a winning strategy $\tau$ of Adam in the letter game on $\A$, and let $\B$ be a deterministic LTS that recognises the plays of $\tau$, seen as runs of $\A$.
	Let $\B_\Sigma$ be the projection of $\B$ onto $\Sigma$: it is otherwise like $\B$, except that its alphabet is $\Sigma$ instead of the transitions $\Delta_\A$ of $\A$ and as a result it has additional nondeterminism. Crucially, every transition in $\B$ is still a transition in $\B_\Sigma$. Given a sequence of transitions $t_0t_1\dots t_i \in \Delta_\A^*$, we call $t''_0t''_1\dots t''_i$ its run in $\B$, which is uniquely defined since $\B$ is deterministic. Note that this sequence of transitions is also a run over the word of $t_0t_1\dots t_i$ in $\B_\Sigma$. This also extends to infinite sequences. Since every run accepted by $\B$ is a play winning for Adam in the letter game over $\A$, their projection onto $\Sigma$ must be in $L(\A)$, so $L(\B_\Sigma)\subseteq L(\A)$.
	
	Now, let $\B'$ be the $1$-token ghost of $\B_\Sigma$, witnessed by a strategy $\strat_1$ of Eve in the game $G_1(\B',\B_\Sigma)$. Assume, towards contradiction, that $\Sim(\A,\B')$ via some strategy $\strat_{sim}$.
	We construct a strategy $\strat$ of Eve in the letter game on $\A$ that is winning against $\tau$, in which Eve plays $\strat_{sim}$ against her imaginary friend Berta in $\Sim(\A,\B')$, who in turn is playing $\strat_1$ against Adam in $G_1(\B', \B_\Sigma)$.
	
	In more detail, Adam begins by playing $\sigma_0$ according to $\tau$ in the letter game on $\A$; 
	Berta responds with a transition (or sequence of transitions in the presence of $\epsilon$-transitions) $\rho'_0=\strat_1(\sigma_0)$; and then Eve responds with $\strat(\sigma_0)=\rho_0=\strat_{sim}(\rho'_0)$. 
	On the $i^{th}$ round, when Adam chooses the letter $\sigma_{i}$, after the sequence $\rho_0\dots \rho_{i-1}$ of Eve's moves in the letter game, and the sequence  $\rho''_0,\dots,\rho''_{i-1}$ of transitions in $\B_\Sigma$, which is the $\Sigma$-projection of the unique run of $\B$ on $\rho_0\dots \rho_{i-1}$, viewed as a word over $\Delta_A$, Berta makes the move $\rho'_i=\strat_1(\rho''_0,\dots, \rho''_{i-1},\sigma_i)$ in $G_1(B',\B_\Sigma)$, and then Eve the move $\strat(\sigma_0,\rho_0,\dots,\rho_{i-1},\sigma_i)=\rho_i=\strat_{sim}(\rho'_0,\dots,\rho'_{i-1})$ in $\Sim(\A,\B')$ and in the letter game.
	
	We argue that $\strat$ is winning against $\tau$. Indeed, the run $\rho''_0\rho''_1\dots$ in $\B_\Sigma$ must be accepting since the sequence of transitions $\rho_0 \rho_1\dots$ that Eve plays agrees with $\tau$. Then, since Berta is playing a winning strategy in $G_1(\B',\B_\Sigma)$, the sequence $\rho'_0 \rho'_1\dots$ is also an accepting run over the same word. Since Eve is playing a winning strategy in $\Sim(\A,\B')$, the sequence $\rho_0 \rho_1\dots$ is also an accepting run over the same word. This contradicts $\tau$ being a winning strategy for Adam.
	We conclude that $\A$ does not simulate $\B'$ and is therefore not guidable.
\end{proof}


\cref{cl:CompositionCriterion} can be applied to various automata classes, as summarised in \cref{table:table1}, including $\omega$-regular automata with an $[i,j]$-parity acceptance condition (\cref{sec:Straightforward}), linear automata (\cref{sec:linear}), and visibly pushdown automata (\cref{sec:vpa}).


This concludes our criteria. Concerning the necessity of each criterion, notice that:
\begin{itemize}
	\item The  first criterion (\cref{cl:Criteria}.\ref{Item:Determinisation}) is not subsumed by the others, as demonstrated with the class of all LTSs---it is closed under determinization ~\cite[Theorem 3.4]{BHLST'23}, but  $G_1$ does not characterise history-determinism, which is required for the second criterion, and the letter game need not always be determined, which is required for the third. 
	\item The second criterion (\cref{cl:Criteria}.\ref{Item:OneTokenGhost}) does not imply the first one, as demonstrated by, for instance, safety pushdown automata~\cite[Theorem 4.1]{GJLZ24}. The implication from the second criterion to third criterion is unclear, however, and connects to the case of PDA, where the strategies for the players in letter game are not yet understood~\cite[Section 6]{GJLZ24}. 
	\item Finally, the third criterion (\cref{cl:Criteria}.\ref{Item:SafetyProduct}) is not subsumed by the other two, as evident from the case of linear automata  (\cref{sec:linear}).
\end{itemize}

\section{Automata Classes for which History-Determinism = Guidability}\label{sec:Yes}

\subsection{Straightforward cases}\label{sec:Straightforward}
By \cref{cl:Criteria}.\cref{Item:Determinisation}, history-determinism and guidability coincide for all automata classes  closed under determinisation, including:  regular automata (NFAs); VPAs on finite words; $\omega$-regular automata \cite{McN66}; co-B\"uchi \cite{MH84}; and subclasses of $\omega$-regular automata whose deterministic fragment is $\omega$-regular-complete, such as parity, Rabin, Streett, Muller, and Emerson-Lei.
%
%
%
Some subclasses of $\omega$-regular automata are not closed under determinisation, e.g., B\"uchi automata, but as long as they subsume safety automata, we can build on the fact that Adam's letter-game strategies are recognised by deterministic safety automata, and apply \cref{cl:Criteria}.\cref{Item:SafetyProduct}: since safety automata are determinisable they are closed under $1$-token ghost.

\begin{corollary}\label{cl: hd-guid parity}
	History-determinism and guidability coincide for classes of $\omega$-regular automata with an $[i,j]$-parity acceptance condition, as well as for the class of weak automata.
\end{corollary}

\subsection{Uniform infinite state systems}
\label{sec:infinite-state}
In this section, we show that the notions of history-determinism and guidability coincide on the following classes with safety and reachability acceptance conditions: pushdown automata, one-counter automata, vector addition system with states, one-counter nets and Parikh automata. We take a unified approach by defining all of these classes as cases of ``uniform automata classes'', and showing that the two notions coincide for such classes (\cref{thm:Simple-automata-hd-vs-guid}).

These uniform automata classes are specified by a content space $\C$ (e.g., stack contents) and a set $\K$ of partial functions $f:\C \rightharpoonup \C$ that contains the identity function $f_{id}$ that maps each element in $\C$ to itself (e.g., stack updates). The class specified by $\C$ and $\K$  contains all the automata $\A = (\Sigma,Q,\iota,c_0,\Delta,F_A,F_C)$ that have a finite alphabet $\Sigma$, a finite state space $Q$, and finitely many transitions $(q,\sigma,f,q') \in \Delta$, labelled by a letter $\sigma \in \Sigma_{\epsilon} = \Sigma \cup \{\epsilon\}$ and a function $f \in \K$. The automaton $\A$ induces an LTS that has states $(q,c) \in Q \times \C$, with transitions $(q,c) \xrightarrow{\sigma} (q',c')$, such that $(q,\sigma,f,q')$ is a transition in $\A$ and $f(c)=c'$. 

The acceptance semantics of an automaton in such a class is specified by a set of accepting states $F_A \subseteq Q$ and a set of accepting contents $F_C \subseteq \C$. We will often desire some structure on $F_C$, so we impose the restriction that $F_C$ belongs to a set $\S \subseteq \P(C)$ of subsets of $\C$. We call ``$(\C,\K,\S)$-automata'' the class of all automata $\A = (\Sigma,Q,\iota,c_0,\Delta,F_A,F_C)$ as above with $F_C \in \S$. 
Safety automata require an accepting run to have all states in $F_A$ and all content in $F_C$. We distinguish between synchronous reachability that requires an accepting run to reach an accepting state and an accepting content at the same time, and asynchronous reachability that requires an accepting run to just reach an accepting state and an accepting content, not necessarily at the same time.

\begin{definition}
   A class of automata is \emph{uniform} if it can be specified as $(\C,\K,\S)$-automata with either safety, synchronous reachability, or asynchronous reachability acceptance semantic.
\end{definition}

We show that uniform automata classes are closed under $1$-token ghost by explicitly constructing for each automaton $\A$ in the class a $1$-token ghost, called $\SigmaShift(\A)$, inspired by Prakash and Thejaswini~\cite[Lemma 11]{PT23}. For each run in $\A$, we will have a run in $\SigmaShift(\A)$ that lags one transition behind. This one-step lag is implemented by storing the previous letter in the state space of $\A$, in addition to the state of $\A$; transitions are then taken based on the previous letter, while reading the current letter, which is now stored.

\begin{definition}[$\SigmaShift$]\label{def:Simpleautomata-SigmaShift}
Let $\A = (\Sigma,Q,\iota,c_0,\Delta,F_Q,F_C)$ be an automaton in a uniform class $(\C,\K,\S)$. The automaton $\SigmaShift(\A) = (\Sigma,Q',\iota',c_0,\Delta',F'_Q,F_C)$ is the $\SigmaShift$ of $\A$, where 
\begin{enumerate}
    \item The set of states $Q'$ is $(Q \times \Sigma) \cup \{\iota'\}$
    \item The set of transitions $\Delta'$ is given by the union of:
    \begin{itemize}
        \item $\{(\iota',\sigma,f_{id},(\iota,\sigma)) \mid \sigma \in \Sigma\}$
        \item $\{((q,\sigma),\sigma',f,(q',\sigma')) \mid \sigma,\sigma' \in \Sigma, \text{ and } (q,\sigma,f,q') \in \Delta\}$
        \item $\{(q,\sigma),\epsilon,f,(q',\sigma) \mid \sigma \in \Sigma, \text{ and } (q,\epsilon,f,q') \in \Delta\}$
    \end{itemize}
    \item The set $F'_Q$ consists of state of the form $(q,\sigma)$ such that $q \in F_Q$, and $\iota'$ if $\iota \in F_Q$. 
\end{enumerate}
The automaton $\SigmaShift(\A)$ has the same acceptance semantics as $\A$ (safety, synchronous reachability or asynchronous reachability).
\end{definition}

\begin{restatable}{lemma}{cllemtgsimpleautomata}\label{lem:1-tg-simpleautomata}
Given an automaton $\A$ in a uniform automata class $C$, the automaton $\SigmaShift(\A)$ is in $C$ and is a 1-token ghost of $\A$.
\end{restatable}

 We show that $G_1$ characterises history-determinism on all uniform automata classes (\cref{lem:G1-for-simple-automata}), by reducing to safety and reachability LTSs~\cite{BHLST'23}. With \cref{lem:1-tg-simpleautomata} and \cref{cl:Criteria}.\cref{Item:OneTokenGhost}, we get that history-determinism and guidability coincide on all uniform automata classes.

\begin{restatable}{theorem}{clsimpleautomatahdvsguid}\label{thm:Simple-automata-hd-vs-guid}
History-determinism and guidability coincide for uniform automata classes.
\end{restatable}

It now suffices to represent various automata classes as uniform ones to show that guidability and history-determinism coincide on them. For pushdown and one-counter automata, vector addition systems and one-counter nets, as well as for Parikh automata, the contents are the counter or stack contents, while the update functions are their increments, decrements, pops and pushes. The update partial functions also implement which parts of the contents can be used to enable transitions: for example, for pushdown automata, the partial functions are either defined for all contents where the stack is empty, or undefined for all such contents; for Parikh automata, the contents do not influence which transitions are enabled, so the functions are fully defined. (Formal definitions are given in \cref{def:PDA,def:VASS,def:Parikh-automata}.)
 
\begin{corollary}\label{cl:hd-guid-pushdownparikhvass}
	History-determinism and guidability coincide for the classes of pushdown automata, one-counter automata, vector addition systems with states, one-counter nets with safety and reachability acceptance conditions, and for Parikh automata with safety, synchronous reachability and asynchronous reachability acceptance conditions.
\end{corollary}

\paragraph*{Non-uniform classes}
The class of visibly pushdown automata is not uniform, as there are additional constraints on transitions, namely the kind of function that changes \content depends on the letter seen. Timed automata also do not constitute a uniform class, since the alphabet is infinite as it consists of all timed letters, and the clock valuations are updated according to both the transition (resets) and the delay of the input letter. In \cref{sec:linear}, we consider linear automata: these are B\"uchi automata that have no cycles apart from self-loops. Linear automata also does not form a uniform class, since they restrict the state-space. 
In what follows, we give alternative constructions of $1$-token-ghosts for these classes. The case of linear automata is trickier, as we show that it is not closed under $1$-token ghost. We therefore use, in~\cref{sec:linear}, a more involved argument that allows us to use \cref{cl:Criteria}.\cref{Item:SafetyProduct}.

\subsection{Visibly pushdown and timed automata}\label{sec:vpa}
Visibly pushdown automata over infinite words ($\omega$-VPAs) are neither (history-) determinisable, nor does $G_1$ characterise history-determinism on them.
Nevertheless, we can use \cref{cl:Criteria}.\cref{Item:SafetyProduct} to show that history-determinism and guidability coincide for this class.
 
\begin{restatable}{theorem}{clhdguidvpa}\label{cl:hd-guid-vpa}
	History-determinism and guidability coincide for the class of visibly pushdown automata with any $\omega$-regular acceptance condition.
\end{restatable}

\begin{proof}[Proof sketch]
First we show that the class is closed under $1$-token ghost. 
Like in the previous cases, we build an automaton that executes the same transitions, but one step later. The technical challenge is executing transitions with a delay, as an $\omega$-VPA must respect the stack discipline of the input alphabet. We overcome this by maintaining a ``semantic stack'' that consists of the actual stack and one additional letter that is embedded in the state space and stores, when necessary, the letter that should have been in the top of the stack.
	
Then, we describe the letter-game for an $\omega$-VPA as a game on a visibly pushdown arena with a ``stair parity'' acceptance condition, to show that Adam's winning strategies can be implemented by $\omega$-VPA transducers. We then turn this transducer into a deterministic $\omega$-VPA recognising the plays that agree with Adam's strategy, and apply \cref{cl:Criteria}.\cref{Item:SafetyProduct}. 
\end{proof}

We turn to safety and reachability timed automata, for which we apply \cref{cl:Criteria}.\cref{Item:OneTokenGhost}, yet with a specially tailored $\Delay$ construction.

\begin{restatable}{theorem}{cltimedhdguid}\label{cl:timedhd-guid}
	For the class of timed automata with safety or reachability acceptance conditions, the notions of history-determinism and guidability coincide.
\end{restatable}
\begin{proof}[Proof sketch]
	The goal is to simulate such an automaton $\A$ with a delay, as in~\cref{def:Simpleautomata-SigmaShift}.
	Yet, the difficulty is that delaying a clock-reset by a step will affect the value of the clock for future comparisons, and there is no delaying of the passage of time. Hence a naive construction would end up recognising the timed language of words in which timestamps are shifted by one. 
	We overcome the difficulty by duplicating in the $1$-token ghost construction each clock of $\A$, using one copy for comparisons in guards and the other to simulate retroactive resets. In addition, the state-space stores the effect of the previous delay, by remembering the corresponding region, that is, how the timestamp compares to existing clocks and constants. With this construction, and the $G_1$-characterisation of history-determinism for safety and reachability automata, we complete the proof.
\end{proof}

\subsection{Linear automata}\label{sec:linear}
A \emph{linear} (also called \emph{very weak}) automaton is a B\"uchi automaton in which all cycles are self loops. (In linear automata, the acceptance condition does not really matter, since over an automaton with only self loops, all the standard $\omega$-regular acceptance conditions coincide.)

First, observe that linear automata are not closed under (history-)determinisation. (The standard B\"uchi automaton over the alphabet $\Sigma=\{a,b\}$ recognizing the language of finitely many $a$'s is linear.)
We show that they are also not closed under $1$-token ghost, by proving that the linear automaton depicted in \cref{fig:Linear} admits no $1$-token ghost in the class.

\begin{restatable}{theorem}{cllinearnodelay}\label{cl:LinearNoDelay}
	The class of linear automata is not closed under  $1$-token ghost.
\end{restatable}

\begin{proof}
	Let $\A$ be the linear automaton depicted in \cref{fig:Linear}, and assume toward contradiction that there is a linear automaton $\A'$, satisfying $\Ghost{1}(\A',\A)$, witnesses by a winning strategy $\strat$ of Eve in $G_1(\A',\A)$.
	
	In a play $\pi_1$ of $G_1(\A',\A)$ in which Eve plays along $\strat$ and Adam plays $(ab)^*$ while staying in $q_0$, at some points of time $2k{-}1$ and $2k$, Eve must remain in the same state $q'$ of $\A'$ after Adam chose the letters $a$ and $b$, respectively, since $\A'$ is linear.
	
	In a play $\pi_2$ of $G_1(\A',\A)$ in which Eve plays along $\strat$ and Adam starts with $(ab)^{k{-}1}a$ while staying in $q_0$, Eve reaches, as per the previous claim, the state $q'$ of $\A'$. Then, if Adam continues with the word $ca^\omega$, while moving from $q_0$ to $q_1$ (over the previous $a$) and then to $q_3$ (over $c$), Eve has some accepting continuation run $\rho$ from the state $q'$ over the suffix $ca^\omega$, since $\strat$ is winning for Eve in $G_1(\A',\A)$ and Adam's run is accepting.
	
	Thus, there is an accepting run of $\A'$ on the word $w=(ab)^{k}ca^\omega$, following in the first $2k$ steps the run of Eve in the play $\pi_1$, reaching the state $q'$, and then in the next steps following her accepting continuation in $\pi_2$. Yet, $\A$ does not have an accepting run on $w$, contradicting the equivalence of $\A$ and $\A'$, and thus the assumption that $\Ghost{1}(\A',\A)$.
\end{proof}

\begin{figure}
	\centering
	\begin{tikzpicture}[thick]
		\tikzset{every state/.style = {minimum size =25}}
		
		\node[state] (0) at (0,0) {$q_0$};
		\node[state] (1) at (2,0.65) {$q_1$};
		\node[state] (2) at (2,-0.65) {$q_2$};
		\node[state,accepting] (3) at (4,0) {$q_3$};
		
		\path[-stealth]
		(-.75,0) edge (0)
		(0) edge[loop above] node[right,xshift=0.1cm] {$a,b$} ()
		(0) edge node[above] {$a$} (1)
		(0) edge node[below] {$b$} (2)
		(1) edge node[above] {$c$} (3)
		(2) edge node[below,xshift=-0.1cm,yshift=-0.1cm] {$d$} (3)
		(3) edge[loop above] node[right,xshift=0.1cm] {$a,b$} ()
		;
	\end{tikzpicture}
	\caption{A linear automaton which admits no linear automaton that is a $1$-token ghost of it.}
	\label{fig:Linear}
	
\end{figure}
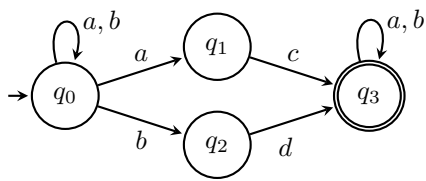

Yet, history-determinism and guidability do coincide for the class of linear automata.
The underlying reason is that when a linear automaton $\A$ is not history-deterministic, Adam's winning strategy in the letter game can be adapted to a linear automaton that does have a $1$-token ghost within the class of linear automata, thus satisfying  \cref{cl:Criteria}.\ref{Item:SafetyProduct}.

\begin{restatable}{theorem}{cllinearhdguidability}\label{cl:LinearHdGuidability}
	History-determinism and guidability coincide for the class of linear automata.
\end{restatable}

\begin{proof}[Proof sketch]
	History-determinism implies guidability with respect to all classes. For the other direction, consider a linear automaton $\A = (\Sigma,Q,\iota,\Delta,\alpha)$ that is not \HD, and let $\M=(\Delta, \Sigma, M,m_0,\Delta_M:M \times \Delta \to M,\gamma:M\to \Sigma)$ be a deterministic finite-state transducer representing a finite-memory winning strategy $\strat_M$ of Adam in the letter game.
	
We then build, by taking a product of $\M$ and $\A$, a deterministic safety automaton  $\P$, that recognises the set of plays that can occur in the letter game on $\A$ if Adam plays according to $\strat_M$. 
From $\P$, we take its projection $\N$ onto the alphabet $\Sigma$ of $\A$. $\N$ need not be linear, but we adapt it into a linear $\N'$ that will still correspond to a winning strategy of Adam in the letter game. $\N'$ will thus constitute a projection of a deterministic automaton $\P'$ onto the alphabet $\Sigma$, where $\P'$ is over the alphabet of transitions of $\A$ and recognises the plays of a winning strategy of Adam in the letter game. Once achieving that, we can apply the $\SigmaShift$ construction on $\N'$---it will not introduce, in this case, non-self cycles, since the states of $\N'$ (as the projection of the states of $\P'$), have outgoing transitions only on a single letter.
Hence, we satisfy \cref{cl:Criteria}.\cref{Item:SafetyProduct}, proving the stated claim.

The transformation of $\N$ into $\N'$ is quite technical, and is left to~\cref{app-sec:linear}.
\end{proof}


\section{Automata Classes for which History-Determinism $\neq$ Guidability }\label{sec:No}

In this section we study classes which admit guidable automata that are not history-deterministic. They offer insight into how, in practice, the criteria can fail to hold, and witness that even on arguably natural automata classes, guidability and history determinism do not necessarily coincide.
The main reason for the equivalence between the notions to fail for these classes is a bound on the allowed resources -- the number of states in the first class and the number of clocks in the second.

Our first example of when history-determinism and guidability differ are B\"uchi automata with a bounded number of states, witnessed by the automaton in~\cref{fig:SmallBuchi}.

\begin{restatable}{theorem}{clBoundedStates}\label{cl:BoundedStates}
	For every $n\in\PNat$, history-determinism and guidability are distinct notions for the class of B\"uchi automata with up to $2n$ states.
\end{restatable}

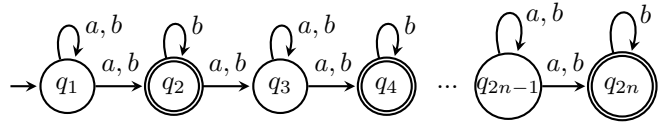
\begin{figure}[h]
	\centering
	\begin{tikzpicture}[thick]
	\tikzset{every state/.style = {minimum size =15}}
	
	\node[state] (i) at (0,0) {$q_1$};
	\node[state, accepting] (1) at (1.4,0) {$q_2$};
	\node[state] (i') at (2.8,0) {$q_3$};
	\node[state, accepting] (1') at (4.2,0) {$q_4$};
	
	\node (m) at (5.0,0) {$...$};
	
	\node[state] (j) at (5.8,0) {$\!\!\!q_{2n{-}1}\!\!\!$};
	\node[state, accepting] (2) at (7.3,0) {$q_{2n}$};
	
	\path[-stealth]
	(-.75,0) edge (i)
	(i) edge node[above] {$a,b$} (1)
	(i) edge[loop above] node[right,xshift=0.1cm] {$a,b$} ()
	(1) edge[loop above] node[right,xshift=0.1cm] {$b$} ()
	(1) edge node[above] {$a,b$} (i')
	(i') edge node[above] {$a,b$} (1')
	(i') edge[loop above] node[right,xshift=0.1cm] {$a,b$} ()
	(1') edge[loop above] node[right,xshift=0.1cm] {$b$} ()
	
	(j) edge[loop above] node[right,xshift=0.1cm] {$a,b$} ()
	(j) edge node[above] {$a,b$}  (2)
	(2) edge[loop above] node[right,xshift=0.1cm] {$b$} ()
	;
	\end{tikzpicture}
	
	\caption{A B\"uchi automaton that accepts words with a finite number of $a$'s. To simulate any equivalent small enough B\"uchi automaton $\B$, Eve moves to the next accepting state once the other automaton is in a maximally strongly connected component with an accepting state. The size constraint on $\B$, and the observation that a such a component can not both have a transition on $a$ and an accepting state guarantees that this strategy wins in the simulation game. However, $\B$ is not history-deterministic.}
	\label{fig:SmallBuchi}
\end{figure}

This counter-example is simple, but quite artificial. We proceed with a class which is, arguably, more natural: timed automata with a bounded number of clocks.

\begin{restatable}{theorem}{clclock}\label{cl:1-clock}
	History-determinism and guidability are distinct notions for the class $\mathbb{T}_k$ of  timed-automata over finite words with at most $k$ clocks, for each $k \in \mathbb{N}$.
\end{restatable}

\begin{proof}[Proof sketch]
	We consider the language of infinite words in which there are $k$ event pairs that occur exactly one time-unit apart both before and after the first occurence of a $\$$ letter. 
	Then, the guidable automaton for this language can freely reset its clocks until the $\$$-separator, which allows it to ensure it tracks all delays tracked by a smaller automaton with up to $k$ clocks. Crucially, any automaton that only accepts words in this language must keep track of $k$ clock values when the separator occurs, as otherwise, it will also accept some word in which the second of matching pair of event is shifted a little.
\end{proof}



\section{Conclusions}\label{sec:conclusions}
We have presented sufficient conditions for a class of automata to guarantee the coincidence of history-determinism and guidability, and used them to show that this is the case for many standard automata classes on infinite words. As a result, we get algorithms to decide guidability for many of these classes. Guidable automata allow for simple model-checking procedures, and once guidability check is simple, one can take advantage of it whenever applicable. For example, consider a specification modelled by a B\"uchi or coB\"uchi automaton $\A$. Model-checking whether a system $\S$ satisfies $\A$ is $\pspace$-hard. Using our results, one can check first in $\ptime$ whether $\A$ is guidable, and in the fortunate cases that it is, conclude the model checking in $\ptime$, by checking whether $\A$ simulates $\S$. We have also demonstrated automata classes for which guidability and history-determinism do not coincide.

We believe that our positive results extend to additional automata classes, such as register automata~\cite{KF94}, which behave quite similarly to timed automata. Furthermore, we believe them to extend to additional families of automata classes:
\begin{itemize}
\item \textit{Finite words.} We have focused on automata over infinite words, which in this context, are better behaved. Ends of words bring additional complications to our constructions, but overall we believe our approach to be amenable to the analysis of finite word automata.
\item \textit{Quantitative automata.}  In quantitative automata, transitions carry additional information in the form of weights. As a result, there is an additional difference between the letter game and simulation game, which makes extending our analysis to the quantitative setting particularly relevant. We believe that many of our techniques adapt to that setting.
\end{itemize}

One could argue that for model-checking, the more interesting property is whether a (not necessarily safety) automaton is guidable by just safety automata, since we typically represent specifications by safety automata. Interestingly, this property often coincides with guidability w.r.t. the full class of automata, as demonstrated in by our third criterion (\cref{cl:Criteria}.\cref{Item:SafetyProduct}): if Adam's strategies in the letter game can be translated into automata, these automata are safety ones, and therefore guidability w.r.t.\ safety automata is just as hard as guidability w.r.t.\ the full class of automata with the more complex acceptance conditions.



\bibliography{gfg}

\newpage
\appendix

\section{Detailed Definitions}\label{app-sec:notation}

\paragraph*{Automata} 
Here we give more detailed definitions of the automata we work with. Each type of automaton induces an LTS of which the states are the configurations of the automaton and the transitions enabled at each configuration updates it as expected; runs and acceptance are defined uniformly from the induced LTS.

A pushdown automaton (PDA) $(\Sigma,Q, \iota, \Delta,\Gamma, \alpha)$ consists of a finite set of states $Q$, a finite input alphabet $\Sigma$, an initial state $\iota$, a transition relation $\Delta$ detailed hereafter, a finite stack alphabet $\Gamma$, and an acceptance condition $\alpha$. 
A configuration consists of a state from $Q$ and a string of stack variables in  $\Gamma^*\bot$ ending in $\bot$, a distinguished bottom-of-stack symbol. The \textit{mode} of a configuration $(q,XS)$, for $q\in Q$, $X\in\Gamma$ and $\S\in\Gamma^*\bot$, is $(q,X)$, and of $(q,\bot)$ is $(q,\bot)$, namely the state and topmost stack symbol. The transitions of a PDA are enabled at a mode and take the form $(q,X)\xrightarrow{\sigma\mid \op} q'$ where $X\in \Gamma\cup \{\bot\}$, $\sigma\in \Sigma\cup \{\epsilon\} $, $\op\in \{\push\ Y,\pop,\noop\}$ for $Y\in \Gamma$, and $q'\in Q$. Transitions update the configuration as follows. A $\push\ Y$ transition adds $Y$ to the top of the configuration stack, a $\pop$ transition removes the top of the stack and can not be enabled at a mode $(q,\bot)$ indicating an empty stack, and a $\noop$-transition leaves the stack as is. A transition changes the configuration state as expected. 

A visibly pushdown automaton (VPA) is a pushdown automaton without $\varepsilon$-transitions of which the input alphabet $\Sigma$ is partitioned between $\push,\pop$ and $\noop$ letters, which only occur in $\push$, $\pop$ and $\noop$ transitions, respectively.

%
A timed automaton $\T = (\Sigma, Q,  \iota, C, \delta, \alpha)$ has a finite alphabet $\Sigma$, a set of states $Q$ with an initial state $q_0$, a set of clocks $C$, a transition function $\delta$ and acceptance condition $\alpha$ as usual. 
Timed automata, as the name suggests, process \emph{timed words}. These are sequences of pairs $\{\sigma_i,d_i\}_{i \geq 0}$ of letters $\sigma_i\in\Sigma$ and non-negative time delays in $\mathbb{R}$. A timed word may also get written as $\{\sigma_i,t_i\}_{i\geq 0}$ where $t_i$ is the timestamp of $\sigma_i$, that is, the sum of the delays up to $i$. We then say that the event $\sigma_i$ occurs at time $t_i$, or after the delay $d_i$. A \emph{timed automaton} is a finite automaton equipped with a set of \emph{clocks} $C$ that measure the passage of time. Each transition $(q,g, \sigma, X,q')$ of a timed automaton checks the valuation of these clocks using a \emph{guard} $g$, and resets the values of a subset of clocks $X$ to 0. These guards are Boolean combinations $\B(C)$ of clock constrains of the form $c \lhd n$, where $c$ is the value of a clock, $n$ is a natural number and $\lhd \in \{\leq,<,=,\geq,>\}$.  

A configuration of a timed automaton is a pair in $Q\times \mathbb{R}^C$ consisting of a state and a clock valuation. A transition $(q,g,\sigma,X,q')$ can be taken over a timed letter $(\sigma,d)$ at configuration $(q,\nu)$ if the valuation $\nu+d$, in which all clocks are increased by $d$ from $\nu$, satisfies the guard $g$. The next configuration is then $(q',\nu+d[X])$, where the clocks in $X$ are reset to $0$.
%
We may use $\pair{q,\nu} \xrightarrow{\sigma,d} \pair{q,\nu+d} \xrightarrow{g,X} \pair{q',\nu'}$ to denote the transition we just described, splitting it into a `delay step' and a `reset step'.
The induced LTS of $\A$ has as its states configurations of $\A$ and as transitions those enabled at each configuration, labelled with letter-delay pairs.

We consider timed automata with safety or reachability acceptance conditions, which we denote by $\T = (\Sigma, Q,  \iota, \delta, F)$ where $F$ is the safe or target region respectively. 

\section{Proofs of Section \ref{sec:criteria}}\label{app-sec:criteria}

\clwhenndethdet*

\begin{proof}
\HD always implies guidability: given $\A$ in $C$ that is \HD, it is easy to see that Eve wins $\Sim(\A,\A')$ for all $\A' \leq \A$, as the strategy that chooses transitions in $\A$ according to the letter game on $\A$, ignoring Adam's token in $\A'$, is a winning strategy.

For the other direction, consider $\A\in C$ that is guidable within $C$, and let $\A'\in C$ be an \HD LTS equivalent to $\A$. As $\A$ is guidable, it simulates, in particular, $\A'$. Eve can use her winning strategy $\strat$ in this simulation game, together with her winning strategy $\strat'$ in the letter game on $\A'$ to also win the letter game on $\A$: when Adam provides a letter in the letter game on $\A$, Eve uses the transition (or transition sequence, in the presence of $\epsilon$-transitions) $\rho$ chosen by $\strat'$ in the letter game on $\A'$ on the same letter. She assumes Adam chooses the transition sequence $\rho$ on $\A'$ in $\Sim(\A,\A')$ as well, and responds to Adam in the letter game on $\A$ with the move suggested by $\strat$ in $\Sim(\A,\A')$. If the word generated by Adam is accepting, then the run generated by $\sigma'$ in $\A'$ is accepting, and therefore the run generated in $\A$ in $\Sim(\A,\A')$ by $\sigma$ is also accepting. 
\end{proof}

\section{Proofs of Section \ref{sec:Yes}}\label{app-sec:yes}

\subsection{Uniform infinite state systems}\label{app-sec:uniform}
\cllemtgsimpleautomata*
\begin{proof}
    Consider $\A = (\Sigma,Q,\iota,\Delta,F_Q,F_C)$ and $\SigmaShift (\A)$~=~$ (\Sigma,Q',\iota,c_0,\Delta',F'_Q,F_C)$ as per \cref{def:Simpleautomata-SigmaShift}. 
    We need to show that 
    (i) Eve wins $G_1(\SigmaShift(\A),\A)$, and 
    (ii) $L(\A) = L(\SigmaShift(\A))$.

(i) We give a winning strategy for Eve in $G_1(\SigmaShift(\A),\A)$. 
In the first round, after Adam chooses a letter $\sigma_0$, Eve chooses the unique transition $(\iota',\sigma_0,f_{id},(\iota,\sigma_0))$, and then Adam a transition sequence $\rho_0 =(\iota,c_0) \xRightarrow{\sigma_0, \rho_0} (q_1,c_1)$. 
Eve will then preserve in next rounds the following invariant:
    \emph{Suppose Adam took the transition sequence  $\rho_i =  (q_{i-1},c_{i-1}) \xRightarrow{\sigma_i,\rho_i} (q_i,c_i)$ in round $i$, then Eve's configuration after $i$ rounds is given by $((q_{i-1},\sigma'_{i-1}),c_{i-1})$}.

 The invariant clearly holds after the first round. Suppose that after $i>0$ rounds of play, the invariant holds. Eve preserves the invariant after $i+1$ rounds as follows:

In the $(i+1)^{th}$ round, suppose Adam selects the letter $\sigma_i \in \Sigma$, after choosing the transition sequence $\rho_i = (p^0,c^0) \xrightarrow{\epsilon,f^0} (p^1,c^1) \xrightarrow{\epsilon,f^1} \cdots (p^k,c^k) \xrightarrow{\sigma_{i-1},f^k} (p^{k+1},c^{k+1} \xrightarrow{\epsilon,f^{k+1}} \cdots (p^l,c^l))$ in the $i^{th}$ round, where $(p^0,c^0) = (q_{i-1},c_{i-1})$ and $(p^l,c^l) = (q_i,c_i)$.

By construction of $\SigmaShift(\A)$, it has a transition sequence $\rho'_i = ((p^0,\sigma_{i-1}),c^0) \xrightarrow{\epsilon,f^0}  \cdots ((p^k,\sigma_{i-1}),c^k) \xrightarrow{\sigma_{i},f^k} ((p^{k+1},\sigma_i),c^{k+1}$ $\xrightarrow{\epsilon,f^{k+1}} \cdots ((p^l,\sigma_i),c^l))$ on $\sigma_i$.
We let Eve pick $\rho'_i$ in $(i+1)^{th}$ round, essentially copying Adam's transition sequence in the $i^{th}$ round. Adam selects a transition sequence $(q_i,c_i) \xRightarrow{\sigma_i,\rho_{i+1}} (q_{i+1},c_{i+1})$, and it is clear that the invariant is preserved throughout. This implies that the run Eve constructs in $\SigmaShift(\A)$ is accepting whenever Adam's is accepting. Hence, Eve wins $G_1(\SigmaShift(\A),\A)$.

(ii)  $L(\A) \subseteq L(\SigmaShift(\A))$ is clear, as Eve wins $G_1(\SigmaShift(\A),\A)$. For the other direction of the inclusion, let $w$ be a word accepted by $\SigmaShift(\A)$ via a run $ \rho' = (\iota',c_0) \xrightarrow{\sigma_0,f_{id}} ((\iota,\sigma_0),c_0) \xRightarrow{\sigma_1,\rho'_0} ((q_1,\sigma_1),c_1) \xRightarrow{\sigma_2,\rho'_1} ((q_2,\sigma_2),c_2)\cdots$, such that $w = \sigma_0  \sigma_1\sigma_2 \cdots$.
Then it is easy to see that there is a run $ \rho = (\iota,c_0) \xRightarrow{\sigma_0,\rho_0} (q_1,c_1)$ $\xRightarrow{\sigma_1,\rho_1} (q_2,c_2) \cdots$ of $\A$ on $w$ that is accepting whenever $\rho'$ is.
\end{proof}

\begin{lemma}\label{lem:G1-for-simple-automata}
    Let $\A$ be an automaton in a uniform automata class. Eve wins $G_1$ on $\A$ if and only if $\A$ is history-deterministic. 
\end{lemma}
\begin{proof}
    Let $\A = (\Sigma,Q,\iota,c_0,\delta,F_Q,F_C)$. When $\A$ has safety or synchronous reachability acceptance semantic,  Eve wins $G_1$ on $\A$ if and only if $\A$ is \HD, since the one-token game characterises history-determinism on all safety and reachability LTSs~\cite{BL22,BHLST'23}.

    For the case of $\A$ with acceptance by asynchronous reachability, consider the reachability LTS $\B$, in which the states are given by $(q,c,p_1,p_2)$, where $q$ is a state in $\A$ and $c \in \C$, and $p_1$ and $p_2$ are flags that are either 0 or 1, based on whether an accepting state in $F_A$ or $F_C$ is seen in the run so far. The initial state is given by $(\iota,c_0,0,0)$. The transitions of $\B$ are of the form $(q,c,p_1,p_2) \xrightarrow{\sigma} (q',c',p'_1,p'_2)$, such that $(q, \sigma, f, q')$ is a transition in $\A$, where $f(c) = c'$, $p'_1 = 1$ if $q' \in F_C$ and $p'_1 = p_1$ otherwise, and similarly, $p'_2 = 1$ if $c' \in F_C$ and $p'_2 = p_2$ otherwise. The accepting states of the LTS $\B$ are those in which both the flags are $1$, i.e., states of the form $(q,c,1,1)$ where $q \in Q$ and $c \in \C$.

    It is easy to see that a winning strategy for Eve in $G_1(\A)$ or the letter game on $\A$ can be used to get a winning strategy for Eve in $G_1 (\B)$ or the letter game on $\B$, respectively, and vice versa. As Eve wins $G_1$ on any reachability LTS if and only if the LTS is \HD~\cite{BL22,BHLST'23}, we get that Eve wins $G_1$ on $\A$ if and only if $\A$ is \HD. 
\end{proof}

We represent various classes of automata as uniform automata classes. The class of safety and reachability nondeterministic $\omega$-regular automata is a uniform automata class, as it can be given by a $(\C,\K,\S)$-automata class where $\C = \{X\}$ is a singleton set, $\K$ only consists of the identity function, and $\S = \{\C\}$. Synchronous and asynchronous reachability coincides in this setting. We show that the classes of VASS,  one-counter nets, and pushdown-, Parikh-, and one-counter-automata are also uniform automata classes.

\begin{definition}\label{def:PDA}
    \emph{Pushdown automata} is a uniform automata class given by $(\C,\K,\S)$, where $\C=\bot\Gamma^{*}$ is the set of possible \emph{stack contents} for some set $\Gamma$ of \emph{stack symbols}, where $\bot \notin \Gamma$. Let $\Gamma_{\bot}$ be $\Gamma \cup \{\bot\}$. The set $\K$ consists of the following partial functions from $\C$ to $\C$:
    \begin{enumerate}
        \item \emph{Push transitions:} Map $\gamma X$ to $\gamma X Y$, where $\gamma X \in \bot\Gamma^{*}$, $X \in \Gamma_{\bot}$ and $Y \in \Gamma$.
        \item \emph{Pop transitions:} Map $\gamma X$ to $\gamma$, where $\gamma \in \bot \Gamma^{*}$ and $X \in \Gamma$.
        \item \emph{Noop transitions:} Map $\gamma$ to $\gamma$ for each $\gamma \in \bot\Gamma^{*}$.
    \end{enumerate}
    The set $\S$ either equals $\{\C\}$, in which case the acceptance is only state-based or equals $\{\{\bot\}\}$, the set consisting of only the `empty stack'. Both these formalisms are known to be equivalent for synchronous reachability PDA, while for safety PDA we consider only the former, as the latter coincides with finite state automata. 
\end{definition}

\emph{One-counter automata} are restrictions of pushdown automata where $\Gamma = \{X\}$ is a singleton set in \cref{def:PDA} above. It is clear that imposing this restriction yields another uniform automata class.

\begin{definition}\label{def:VASS}
  \emph{Vector addition systems with states} (VASS) is a uniform automata class given by $(\C,\K,\S)$, where $\C=\mathbb{N}^{*}$ is the set of \emph{vectors}. The set $\K$ consists of partial functions specified by $\alpha \in \mathbb{Z}^{d}$ for some $d>0$ that maps a vector $v \in \mathbb{N}^{d}$ to $v+\alpha$, provided $v+\alpha$ is in $\mathbb{N}^{d}$. The set $\S$ equals $\{\C\}$, and thus acceptance is only based on the accepting states of an automaton in this class. The accepting condition is given by (synchronous) reachability or safety semantics. Note that as synchronous reachability and asynchronous reachability coincide here due to $\S$ consisting of only $\C$, we just use reachability VASS to denote synchronous reachability VASS. 
\end{definition}
\emph{$d$-dimensional} VASS are restrictions of VASS where $\C = \mathbb{N}^d$ for some $d>0$. It is easy to see that $d$-dimensional VASSes are also a uniform automata class. The class of $1$-dimensional VASS is called \emph{one-counter nets} (OCNs).

\begin{definition}\label{def:Parikh-automata}
    \emph{Parikh automata} is a uniform automata class given by $(\C,\K,\S)$ where $\C=\mathbb{N}^{*}$ is the set of \emph{vectors}. The set $\K$ consists of partial functions specified by $\alpha \in \mathbb{N}^d$ for some $d>0$, that maps a vector $v \in \mathbb{N}^d$ to $v+ \alpha$. The set $\S$ is the union of all semilinear $\S_d$, where $\S_d$ consists of all the semilinear sets in $\mathbb{N}^d$. We consider acceptance conditions based on safety, synchronous reachability and asynchronous reachability semantics.
\end{definition}


\subsection{Visibly pushdown automata}\label{app-sec:vpa}

\begin{lemma}\label{cl:VPA-ghost}
The class of $\omega$-VPAs is closed under $1$-token ghost.
\end{lemma}

\begin{proof}
Given an $\omega$-VPA $\A=(Q, \Sigma, \iota,\Delta,\Gamma, \alpha)$, we build an $\omega$-VPA $\A'$ that is a $1$-token ghost of $\A$, by morally delaying the transitions by one step, as in the previous $\SigmaShift$ constructions.
However, there is a difficulty in doing that, as $\A'$ must operate on the stack at the same rhythm as $\A$, rather than one step behind.
In particular, when $\A'$ reads a $\push$-letter, it has to push something onto the stack before it ``knows'' what to push, and when reading a $\pop$- or $\noop$-letter after a block of $\push$-letters, it ``wants'' to push something to the stack, as $\A$ did on the previous step, but it cannot.

To manage these local asynchronicities, we maintain a ``semantic stack'' in $\A'$ that consists of the actual stack and one additional letter that is embedded in the state space and stores, when necessary, the letter that should have been in the top of the stack. 
At the bottom of the semantic stack there are two symbols, $\bot \bot'$, rather than just $\bot$, so that when the semantic stack has a single element, stored in the state-space, $\A'$ is still able to read a $\pop$-letter.

Consider a run $r$ of $\A$ on a word $w$, and let the stack of $\A$ after reading the $i$th letter of $w$ be $\bot x_1 x_2 \cdots x_k$. Then the semantic stack of $\A'$ after reading the $i+1$ letter of $w$ is $\bot \bot' x_1 x_2 \cdots x_k$. If the $i+1$ letter is a $\push$-letter, then the actual stack is the same as the semantic one, and the dedicated state-space letter that ends the semantic stack is $\varepsilon$.  If the $i+1$ letter is a $\pop$- or $\noop$-letter, then the actual stack is $\bot \bot' x_1 x_2 \cdots x_{k-1}$, and the dedicated state-space letter is $x_k$.
More formally, $\A'$ has:
\begin{itemize}
\item States $Q\times (\Sigma\cup \varepsilon) \times (\Gamma \cup \{\varepsilon,\bot'\})$, where $(q,\sigma,X)$ indicates the corresponding state of $\A$, the last input letter (or, initially, $\varepsilon$), and the dedicated state-space for the top of the semantic stack;
\item Alphabet $\Sigma$;
\item Initial state $(\iota,\varepsilon,\bot')$;
\item Stack alphabet $\Gamma\cup \{\bot'\}$ 
\item Transitions as detailed below.
\item Acceptance condition $\alpha$, where treating a state $(q,\sigma,X)$ of $\A'$ as the state $q$ of $\A$. 
\end{itemize}

The transitions of $\A'$ are the following, where from each state we consider the cases of $\sigma$, the previous letter, and $\sigma'$, the current letter, being $\push$-, $\pop$-, or $\noop$-letters:
\begin{itemize}
\item From the initial state:
\begin{itemize}
\item $((\iota,\varepsilon,\bot'),\bot) \xrightarrow{\sigma'\mid \push \bot'} (\iota,\sigma',\varepsilon)$;
\item No pop from an empty stack;
\item $((\iota,\varepsilon,\bot'),\bot) \xrightarrow{\sigma'\mid \noop} (\iota,\sigma',\bot')$.
\end{itemize}

\item After a push, that is, for $\sigma$ a $\push$-letter:
\begin{itemize}
	\item For each transition $(q,X)\xrightarrow{\sigma\mid \push Y}q'$ in $\A$ (where  $X\in \Gamma$):
	\begin{itemize}
		\item $((q,\sigma,\varepsilon),X) \xrightarrow{\sigma'\mid \push Y} (q',\sigma',\varepsilon)$;
		\item $((q,\sigma,\varepsilon),X) \xrightarrow{\sigma'\mid \pop} (q',\sigma',X)$;
		\item $((q,\sigma,\varepsilon),X) \xrightarrow{\sigma'\mid\noop}(q',\sigma',Y)$.
	\end{itemize}

	\item For each transition $(q,\bot)\xrightarrow{\sigma\mid \push Y}q'$ in $\A$:
	\begin{itemize}
		\item $((q,\sigma,\varepsilon),\bot') \xrightarrow{\sigma'\mid \push Y} (q',\sigma',\varepsilon)$;
		\item $((q,\sigma,\varepsilon),\bot') \xrightarrow{\sigma'\mid \pop} (q',\sigma',\bot')$;
		\item $((q,\sigma,\varepsilon),\bot') \xrightarrow{\sigma'\mid\noop}(q',\sigma',Y)$.
	\end{itemize}
\end{itemize}

\item After a pop or noop, that is, for $\sigma$ a $\pop$- or $\noop$-letter:
\begin{itemize}
	\item For each transition $(q,X)\xrightarrow{\sigma\mid \pop}q'$ or $(q,X)\xrightarrow{\sigma\mid \noop}q'$ in $\A$, $X\in \Gamma$, $Y\in\Gamma\cup\{\bot'\}$:
	\begin{itemize}
		\item $((q,\sigma,X),Y) \xrightarrow{\sigma'\mid \push X} (q',\sigma',\varepsilon)$;
		\item $((q,\sigma,X),Y) \xrightarrow{\sigma'\mid \pop} (q',\sigma',Y)$;
		\item $((q,\sigma,X),Y) \xrightarrow{\sigma'\mid\noop}(q',\sigma',X)$;
	\end{itemize}
	 \item For each transition $(q,\bot)\xrightarrow{\sigma\mid \noop}q'$ in $\A$:
	 \begin{itemize}
	 	\item $((q,\sigma,\bot'),\bot) \xrightarrow{\sigma'\mid \push \bot'} (q',\sigma',\varepsilon)$;
		\item No pop from an empty stack;
	 	\item $((q,\sigma,\bot'),\bot) \xrightarrow{\sigma'\mid\noop}(q',\sigma',\bot')$.
	 \end{itemize}
 \end{itemize}
 \end{itemize}

Observe that there is bijection between runs of $\A$ and corresponding runs in $\A'$, where the run in $\A'$ is one letter behind and the stack of $\A$ is simulated by the semantic stack of $\A'$. Therefore the two automata recognise the same language.

Furthermore, $G_1(\A',\A)$ since Eve can first take the dummy transition from the initial state and then copy Adam's transitions to obtain a run that is accepting whenever Adam's run is accepting.
We conclude that  $\A'$ is a $1$-token ghost of $\A$.
\end{proof}

We proceed to show that history-determinism and guidability coincide for the class of $\omega$-VPA. We shall do this using the criterion of \cref{cl:Criteria}.\ref{Item:SafetyProduct}, that requires us to construct a safety-VPA recognising the plays of a winning strategy of Adam in the letter game on an $\omega$-VPA, if the $\omega$-VPA is not history-deterministic. Towards this, we use the so-called \emph{stair-parity VPAs} as an intermediary.

\paragraph*{Stair-parity automata}

While deterministic parity VPAs are not as expressive as nondeterministic parity VPAs, it is nevertheless possible to determinise nondeterministic parity VPAs into deterministic \emph{stair-parity} VPAs. A stair-parity VPA has the same components as a parity VPA, just that the acceptance condition is different for run. 

Given a run $\rho$ of a stair-parity VPA, consider the sequence $\rho'$ obtained by discarding configurations for which a future configuration of smaller stack-height exists. Observe that this sequence $\rho'$ must be infinite. We write that the run $\rho$ satisfies the stair-parity condition if $\rho'$ satisfies the parity condition. A stair-parity VPA $\A$ accepts a word $w$ iff there is a run of $\A$ on $w$ that satisfies the stair-parity condition. 

We use the following two known results on stair-parity VPAs.
\begin{lemma}\label{lemma:VPA-stair-parity}
\begin{enumerate}
    \item Every nondeterministic parity VPA can be determinised to a language-equivalent deterministic stair-parity VPA~\cite[Theorem 1]{LMS04}.
    \item Let $\G$ be a two-player turn-based game on a potentially infinite arena between Adam and Eve whose winning condition for Eve is given by a deterministic stair-parity VPA~$\A$. Then, the winning player has a winning strategy which can be implemented by a deterministic visibly pushdown transducer $\M$, that takes as input the transitions $\Delta_A$ of $\A$. Furthermore, the transition on a \push \ (resp. \pop \ and \noop) letter of $\A$ in $\Delta_A$ is a \push \ (resp. \pop \ and \noop) letter of $\M$~\cite[Theorem 9]{Fri10}.    
\end{enumerate}
\end{lemma}

\clhdguidvpa*

\begin{proof}
	We will use the criterion of~\cref{cl:Criteria}.\cref{Item:SafetyProduct}.	

	First, note that an $\omega$-VPA $\B$ with any $\omega$-regular acceptance condition $\beta$ can be replaced by an $\omega$-VPA $\A$ with a parity condition, such that $\A$ accepts the same language as $\B$, and furthermore, the nondeterministic choices in $\A$ are exactly those of $\B$: we construct $\A$ by taking the product of $\B$ with a deterministic parity automaton whose alphabet is the transitions of $\B$, and which recognises the acceptance condition $\beta$. Such a deterministic parity automaton always exists; see~\cite{Bok18}, for instance. 
	As $\A$ and $\B$ are equivalent and have the same nondeterministic choices, the letter games on them are isomorphic, and so are the simulation games and token games.
	We can thus proceed with the proof, assuming that the considered $\omega$-VPA has a parity acceptance condition.
	
	We will describe the letter-game on an $\omega$-VPA with parity acceptance as a game with a finite arena and an $\omega$-VPA winning condition to show that Adam's winning strategies can be implemented by visibly pushdown transducers. We will then turn this transducer into an $\omega$-VPA recognising the plays that agree with Adam's strategy, to then apply \cref{cl:Criteria}.\cref{Item:SafetyProduct}.

	The arena consists of the product of the state-space of $\A$ with the input alphabet $\Sigma$, so that at every other turn, Adam chooses a letter, and every other turn, Eve chooses a matching transition. It will be technically convenient to only label Eve's turns with the transition chosen by Eve, which includes the letter played by Adam. A play then consists of the sequence of transitions played by Eve, from which we can also read the word read by Adam.
	Adam's winning condition is the intersection between $\{w \St $ the word component of $w$ is in $L(\A)\}$ and the set of Eve's runs that are rejecting or invalid. We argue that this is recognised by an  $\omega$-VPA that is synchronised with $\Sigma$.  
	To check whether Eve's run is rejecting or invalid, we take a copy $\D$ of $\A$ that reads the transitions chosen by Eve and deterministically simulates Eve's run. It then accepts if either Eve chooses an invalid transition ($\pop$ on an empty stack or a transition that is not enabled at the right mode) or if the resulting run fails the parity condition (which we can easily complement). Then, let $\A'$ be the intersection of $\A$ and $\D$. Since $\omega$-VPAs are closed under intersection~\cite[Theorem 6]{AM04}, the result is still an $\omega$-VPA. By construction, it recognises Adam's winning condition.
	
	The letter game on an $\omega$-VPA can therefore be seen as a finite-arena game with an $\omega$-VPA winning condition for Adam, which can be determinised into an equivalent deterministic stair-parity VPA (\cref{lemma:VPA-stair-parity}).
 
	 Then, the game can be seen as a game played on the LTS of a deterministic $\omega$-VPA with a stair parity winning condition. In such games, the winning player has a winning strategy implemented by a deterministic visibly pushdown transducer $\M$, of which the alphabet consists of the transitions $\Delta_\A$ of $\A$, and $\push,\pop$ and $\noop$ transitions of $\A$ are $\push,\pop$ and $\noop$ letter for $\M$, respectively~\cite[Theorem 9]{Fri10}. In other words, projected onto the letters of the transitions, this alphabet is synchronised with $\Sigma$.
	
	We can now interpret this transducer $\M$ as a safety $\omega$-VPA $\B$ that recognises the plays of this winning strategy for Adam in the letter game of $\A$, by interpreting the outputs of the transducer as part of the input. In other words, the automaton $\B$, with alphabet $\Delta_\A\times \Sigma$ has the same state-space as $\M$ and transitions $q\xrightarrow{(t,a)}q'$ if $\M$ has transition $q\xrightarrow{t:a}q'$ that produces $a$ when reading $t$ at state $q$. To deal with the initial output of $\M$, the initial state of $\B$ only has outgoing transitions over $(t,a)$ such that $t$ is a transition of $\A$ over the initial output of $\M$. Then, $\B$, seen as a safety automaton, recognises the set of plays of the letter game consistent with Adam's strategy given by $\M$.
	
This automaton $\B$ is then as required for the application of 
criterion~\cref{cl:Criteria}.\cref{Item:SafetyProduct}, since $\omega$-VPAs are closed under $1$-token ghost from \cref{cl:VPA-ghost}. We thus conclude that guidability and history-determinism coincide for $\omega$-VPAs.
\end{proof}



\subsection{Timed automata}\label{app-sec:timed}


We define a $\SigmaShift$ operation for timed automata and show that $\SigmaShift(\T)$ is a 1-token ghost of the automaton $\T$ (\cref{cor:timed-automata-closed-under-delay}). As the one-token game characterises history-determinism on safety/ reachability timed automata~\cite{BHLST'23}, we get by \cref{cl:Criteria}.\cref{Item:OneTokenGhost} that history-determinism and guidability coincide for them.

The timed alphabet is infinite and therefore unsuitable for a product construction. Hence, we have to store the \textit{region} that the previous timed letter produces in the state space. Furthermore, because time does not wait for our one-letter lag to advance, we duplicate each clock, using one copy for comparisons in guards while the other is reset at each transition, just in case the chosen transition actually resets this clock. The copy that is reset at each transition is upgraded to be the clock used to check guards, while the other one is now reset at each transition. 

We first describe the regions formally. For each clock $x \in C$, let $c_x$ be the largest constant in any guard involving $x$ in $\T$. Let $R_x$ be the set consisting of the following intervals in $\mathbb{R}$: 
\begin{itemize}
	\item The points $i$, for each $i \in \mathbb{N}$ with $0 \leq i \leq c_x$.
	\item The open intervals $(i,i+1)$ for each $i \in \mathbb{N}$ with $0 \leq i < c_x$.
	\item The open interval $(c_x,\infty)$, which consists of all values greater than $c_x$.
\end{itemize} 
We let $\R = \prod_{x  \in C} R_x$ to be the set of regions. For each region $r \in \R$, let $g_r \in \B(C)$ be a guard such that a valuation $\nu$ belongs to the region $r$ if and only if $\nu$ satisfies the guard $g_r$, i.e. $\nu \models g_r$. Finally, we say a region $r$ satisfies a guard $g$, denoted by $r \models g$ if and only if all valuations $\nu \in r$ satisfy $g$. 




As mentioned earlier, for each clock $x$ in $\T$, we will have two clocks in $\SigmaShift(\T)$ that we denote by $(x,0)$ and $(x,1)$.  Thus, the clocks in $\SigmaShift(\T)$ are given by $C \times \{0,1\}$. The states in $\SigmaShift(\T)$, apart from an additional initial state, are of the form $(q,\sigma,r,f) \subseteq Q \times \Sigma \times \R \times \{0,1\}^{C}$, where $f: C \xrightarrow{} \{0,1\}$ is a function. This is used to encode which clock amongst $(x,0)$ and $(x,1)$  is \emph{active} in $\SigmaShift(\T)$ at any given moment,  for each clock $x$ in $\T$. For a valuation $\mu: C \times \{0,1\} \xrightarrow{} \mathbb{R}$ at a state $(q,\sigma,r,f)$ in $\SigmaShift(\T)$, the \emph{active valuation} $\restr{\mu}{f} = \mu (x,f(x))$ is the values of the clocks of the form $(x,f(x))$. The clocks that are not active are called \emph{passive}, and these are of the form $(x,1-f(x))$.

We also extend each guard $g \in \B(C)$ to guards for the active clocks in $C \times\{0,1\}$: Let $g[f] \in \B(C \times \{0,1\})$ be the guard obtained by replacing each instance of a clock $x$ in $g$ by $(x,f(x))$. Finally, to deal with resets of clocks $X\subseteq C$, we need to switch the active and passive clocks in $\SigmaShift(\T)$ amongst the clocks that correspond to $X$. This is done by the function $f^X$ obtained by swapping the active and passive clocks associated to $X$ in $f$. That is, $f^X(c) = f(c)$ if $c \notin X$, and $f^X(c) = 1-f(c)$ if $c \in X$.


Let us denote the active and passive clocks with respect to a function $f$ by $A[f]$ and $P[f]$, respectively.
If the set of clocks reset in the $i^{th}$ transition in $\T$ is $X$, then the following clocks would be reset in the corresponding $(i+1)^{th}$ transition in $\SigmaShift(\T)$:
\begin{enumerate}
    \item The passive clocks $\{(c,1-f(c)) \St c\notin X\}$ not in $X$.
    \item The active clocks $\{(c,f(c)) \St x\in X\}$ of $X$.
\end{enumerate}

We now formally define the construction of $\SigmaShift(\T)$ described above:
\begin{definition}
Given a timed automaton $\T = (\Sigma,Q, q_0, C, \delta, F)$, we define the timed automaton $\SigmaShift(\T)  = (\Sigma,Q',s,C,\delta',F')$, where the set of states $Q'$ consists of 
\begin{itemize}
    \item  The initial state, $s$
    \item Tuples $(q,\sigma,r,f) \in Q \times \Sigma \times \R \times \{0,1\}^C $
\end{itemize}  
The set $\delta'$ consists of the following transitions: 
\begin{itemize}
    \item The transitions $(s,\sigma,g_r[f_0], P[f_0], (q_0,\sigma,r,f_0))$ for each $\sigma \in \Sigma, r \in R$, where $f_0$ is the function $f: C \xrightarrow{} \{0,1\}$ such that $f(x) = 0$ for all $x \in C$.
    \item The transitions $((q,\sigma,r,f),$ $\sigma',g_{r'}[f^X], P[f^X], (q',\sigma',r',f^X))$ such that  $r\models g$, for every transition $(q, \sigma, g, X, q')$ in $\T$.
\end{itemize}
The set $F'$ consists of states of the form $(q,\sigma,r,f)$ where $q \in F$ is an accepting state in $\T$, as well as the state $s$ if $q_0$ is in $F$.
\end{definition}

We show next, via the following two lemmas, that $\SigmaShift(\T)$ is a one-token ghost of $\T$.

\begin{lemma}\label{lemma:timed-automata-g1-sigmashift}
Let $\T = (\Sigma,Q, q_0, C, \delta, F)$ be a safety/reachability timed automaton. Then, Eve wins $G_1(\SigmaShift(\T),\T)$.
\end{lemma}
\begin{proof}
We show that Eve can copy the $i^{th}$ transition of $\T$ in her $(i+1)^{th}$ transition of $\SigmaShift(\T)$, to establish that she wins the one-token game between $\SigmaShift(\T)$ and~$\T$.

In the first round, after Adam selects a letter $(\sigma_0, d_0)$, Eve's transition in $\SigmaShift(\T)$ is  deterministic: $(s,\sigma_0, g_r[f_0],P[f_0], (q_0, \sigma_0, r_0,f_0)$, where $r_0$ is the unique region of the valuation $\nu: C \xrightarrow{} \mathbb{R}$ which maps all clock values to $d_0$. Then, Adam would select some transition  $(q_0,\sigma_0,g_0,X, q_1)$ in $\T$, resulting in valuation $\nu_0$.

After $i>0$ rounds, suppose Eve's token in $\SigmaShift(\T)$ is in state $p_{i}=(q_{i-1},\sigma_{i-1},r_{i-1},f_{i-1})$, with valuation $\mu_i$, and Adam's token in $\T$ after $i-1$ rounds is at the state $q'$ with valuation $\nu_{i-1}$. We will show that Eve has a strategy to preserve the following invariant: 
\begin{enumerate}
    \item[\textbf{I1}] Adam's state $q'$ after $i$ rounds is $q_{i-1}$, same as the component of $Q$ in $p_i$.
    \item[\textbf{I2}] 
     The active valuation $\restr{\mu_{i-1}}{f_{i-1}}$ equals $\nu_{i-1}+d_{i}$, while all the passive clocks in $\mu_i$ are assigned the value $0$.
    \item[\textbf{I3}] The region $r_{i-1}$ is the unique region of the active valuation $\restr{\mu_{i-1}}{f_{i-1}}$. 
\end{enumerate}       

 It is clear that the invariant holds after the first round. Suppose it holds after $i>0$ rounds. We will show that Eve has a strategy to  preserve the invariant after $i+1$ rounds:

Suppose, after $i$ rounds, Eve's configuration is given by $\pair{p,\mu}=\pair{(q,\sigma,r,f),\mu}$, while Adam's transition in the $i^{th}$ round was $\pair{q,\nu} \xrightarrow{\sigma,d} \pair{q,\nu+d} \xrightarrow{g,X} \pair{q',\nu'}$. By the invariant, we know that $\restr{\mu}{f} = \nu + d$,  and $r$ is the region of $\restr{\mu}{f}$.  
In the $(i+1)^{th}$ round of the play: 
\begin{enumerate}
    \item Adam selects a timed letter $(\sigma',d')$.
    \item Eve the transition 
    $\pair{(q,\sigma,r,f),\mu} \xrightarrow{\sigma',d'}  \pair{(q,\sigma,r,g),\mu+d'}$ $\xrightarrow{g[f^X],P[f^X]} \pair{(q',\sigma',r',f^X),\mu'}$. Note that this transition exists as, by the invariant, $r$ is the region of $\restr{\mu}{f} = \nu+d$, which satisfies the guard $g$ due to Adam being able to take his transition in the $i^{th}$ round.
    \item Adam a transition $\pair{q',\nu'} \xrightarrow{\sigma',d'} \pair{q',\nu'+d'} \xrightarrow{g',X'} \pair{q'',\nu''}$.
\end{enumerate}

It is clear that the invariant $\textbf{I1}$ is satisfied after the $(i+1)^{th}$ round. 

As for the invariant $\textbf{I2}$, by the induction hypothesis, we know that $\restr{\mu}{f}=\nu+d$. We need to show that $\restr{\mu'}{f^X} = \nu'+d'$.
Note that $\nu'+d'= (\nu+d)[X] + d'$, which is, $\nu(c)+d+d'$ if $c \notin X$, and $d'$ if $c \in X$.
 Let us take a look at how the valuation $\mu'$ is obtained. First, we add the duration $d'$ to the valuation $\mu$. This causes all active clocks with respect to $f$ to have the value $\nu(c)+d'+d$, and the passive clocks to have the value $d'$. Then, the function $f^X$ swaps active and passive clocks of $X$ and resets the clocks that are now passive with respect to $f^X$ to obtain $\mu'$. Thus,  we have that $\mu'(c,f^X(c)) = \mu(c,f(c))+d'=\nu(c)+d+d' = [\nu'+d'](c) \text{ if $c \notin X$} $, and $\mu'(c,f^X(c))= d' = [\nu'+d'](x)$ if $c \in X$. This implies that $\restr{\mu'}{f^X} = \nu'+d'$. It is clear that the passive clocks with respect to $f^X$ have value 0, as they were just reset. Thus, the condition \textbf{I2} is also an invariant.

The invariant \textbf{I3} holds as $\restr{\mu'}{f^X}$ must have satisfied the guard $g_{r'}[f^X]$ to have taken the transition, which implies that $r'$ is the unique region of $\restr{\mu'}{f^X}$. 

Having shown that Eve has a strategy which preserves the invariants \textbf{I1}, \textbf{I2} and \textbf{I3}, it is easy to see that such a strategy is winning in $G_1(\SigmaShift(\T),\T)$: By the invariant \textbf{I1} and the construction of the accepting states of $\SigmaShift(\T)$, Eve's run in $\SigmaShift(\T)$ is accepting if and only if Adam's run in $\T$ is accepting.

\end{proof}

We proceed with showing that $L(\SigmaShift(\T)) = L(\T)$.
\begin{lemma}\label{lem:TA-sigmashift-le}
  Given a timed automaton $\T$, we have $L(\SigmaShift(\T)) = L(\T)$.
\end{lemma}
\begin{proof}
 $L(\T) \subseteq L(\SigmaShift(\T))$ is clear by~\cref{lemma:timed-automata-g1-sigmashift}.
 As for $L(\SigmaShift(\T)) \subseteq L(\T)$, consider a sequence $(\pair{p_i,\mu_i})_{i \geq 0}$ of states of $\SigmaShift(\T)$ that forms a run on the timed word $(\sigma_i,d_i)_{i \geq 0}$. Let $q_j$ be the state component of $p_{j+1}$, and let $\nu_j: C \xrightarrow{} \mathbb{R}$ be obtained by subtracting $d_j$ from $\restr{\mu_{j+1}}{f_{j+1}}$, the active valuation of $\pair{p_j,\mu_{j+1}}$, for each $j \geq 0$. Then, it is easy to verify from construction that the sequence $(\pair{q_i,\nu_i})_{i \geq 0}$ is a run in $\T$ on the timed word $(\sigma_i,d_i)_{i \geq 0}$ which is accepting if and only if $(\pair{p_i,\mu_i})_{i \geq 0}$ is an accepting run in $\SigmaShift(\T)$.
\end{proof}

From~\cref{lemma:timed-automata-g1-sigmashift,lem:TA-sigmashift-le}, we can conclude:
\begin{lemma}\label{cor:timed-automata-closed-under-delay}
The classes of safety and reachability timed automata are closed under $1$-token ghost.
\end{lemma}

As $G_1$ characterises history-determinism on safety and reachability timed automata~\cite{BHLST'23}, by \cref{cl:Criteria}.\cref{Item:OneTokenGhost}  history-determinism and guidability coincide:  
\cltimedhdguid*

\subsection{Linear automata}\label{app-sec:linear}
\cllinearhdguidability*
\begin{proof}
	History-determinism implies guidability with respect to all classes. For the other direction, consider a linear automaton $\A = (\Sigma,Q,\iota,\Delta,\alpha)$ that is not \HD, and let $\M=(\Delta, \Sigma, M,m_0,\Delta_M:M \times \Delta \to M,\gamma:M\to \Sigma)$ be a deterministic finite-state transducer representing a finite-memory winning strategy $\strat_M$ of Adam in the letter game.

    Let $\P = (\Delta,M \times Q, p_0, \Delta_P)$ be the deterministic safety automaton that is obtained by the composition of $\M$ with $\A$. It's states are $Q \times M$, we have the transition $(m,q) \xrightarrow{\delta} (m',q')$ if $\gamma(m) = \sigma$, $\delta = q \xrightarrow{\sigma} q'$ is a transition in $\A$, and $m \xrightarrow{\delta} m'$ is the unique transition in $\M$. Note that $\P$ recognises the set of plays that can happen in the letter game on $\A$ if Adam plays according to $\strat_M$. 
    
    Let $\N=(\Sigma,N,n_0,\Delta')$ be the safety automaton that results from taking the projection of $\P$ to the alphabet $\Sigma$ of $\A$. Note that $\N$ need not be linear. For each $q\in Q$, let $N_q$ be states in $\N$ of the form $(m,q)$ for some $m\in M$. Note that there is no cycle in $\N$ that visits $N_q$ and $N_{q'}$, for two different states $q,q'\in Q$, as $\A$ is linear. 
	Hence any cycle in $\N$ must be contained entirely in one of the $N_q$'s. Let $\N_q$ be the transition system that consists of all states in $N_q$ and transitions between them.
	
	We proceed by adapting $\N$ to a linear automaton $\N'$ that will still correspond to a winning strategy of Adam in the letter game, and will thus constitute a projection of a deterministic automaton $\P'$ onto the alphabet $\Sigma$, where $\P'$ is over the alphabet of transitions of $\A$ and recognises the plays of a winning strategy of Adam in the letter game. Once achieving that, we can apply the $\SigmaShift$ construction on $\N'$ -- it will not introduce, in this case, non-self cycles, since the states of $\N'$ (as the states of $\N$), have outgoing transitions only on a single letter.
	Hence, we satisfy \cref{cl:Criteria}.\cref{Item:SafetyProduct}, proving the stated claim.
	
	We elaborate below on how we adapt $\N$ into $\N'$. 
	
	Observe that since $M$ is a winning strategy for Adam in the letter game, every word that $\N$ reads is accepting in $\A$, and every run of $\N$ eventually remains in a component $\N_q$, where $q$ is a rejecting state of $\A$.
	
	Therefore, if a component $\N_q$ contains a cycle, then $q$ must be a rejecting state of $\A$.
	Furthermore, note that every state of $\N$, and thus of $\N_q$, has outgoing transitions on a single alphabet, and that $\N_q$ is deterministic, as $\A$ has at most one self-loop at the state $q$ on any letter $\sigma$.
	Hence, from each state $n$ of $\N_q$, there is at most one infinite word $w_n$ that can be generated within $\N_q$. 
	The following claim will be the key to linearising $\N$ into $\N'$.
	
	{\bf Claim}: there is a number $K \in \mathbb{N}$, such that for every state $n$ that lies within a cycle of $\N_q$ and finite word $u$ over which $\N$ can reach $n$, there is a run $\rho$ in $\A$ on the finite word $u\cdot w[0..K]$ such that $\rho$ ends in an accepting state $f$ of $\A$, and $f$ has self loops on all letters that appear in the word $w[K{+}1.. \infty]$. 
	\begin{proof}[Proof of the claim]
		We will show that the result holds for $K =  |N_q| \cdot 2^{|Q|} $. 
		Let $v$ be the unique shortest finite word whose unique run in $\N_q$ starts at $n$ and ends at $n$. Note that the length of the word $v$ is at most $|N_q|$. 
		For every number $i\in\Nat$, let $X_i$ be the set of $\A$'s states that can be reached upon reading the finite word $uv^i$ in $\A$. 
		
		By the pigeonhole principle, two of the subsets $X_i$ and $X_j$ for $i<j$ must be equal. Let $d$ be the difference between $j$ and $i$. It follows that the sequence $\{X_l\}_{l\geq 0}$ is eventually periodic after $i$ with period at most $d$, that is $X_{i+k} = X_{i+k+d}$ for all $k\geq 0$. As the word $uv^{\omega}$ is accepting, there is an accepting state $f$ of $\A$ reached by some $uv^{l}$ that has a self-loop on all the letters of $v$. This state $f$ must belong to all of $X_l$'s for all $l$ large enough. As the sequence $\{X_k\}_{k \geq i}$ is periodic, it follows that $f \in X_k$ for all $k \geq i$. 
		
		As the length of $v^i$ is at most $N_q \times 2^{|Q|}$, the state $f$ is thus reached upon reading the finite word $w[0 .. K]$, and has self-loops on all letters in the word $w[K{+}1 .. \infty]$.
	\end{proof}
	
	We now use the claim for generating the required automaton $\N'$ from $\N$, and thus completing the proof.
	We shall modify $\N_q$ to have $K+1$ copies of the states in $\N_q$, that is $N'_q = N_q \times [0..K]$. Then, for each transition $n \xrightarrow{\sigma} n'$ in $\N_q$, we have the transition $(n,i) \xrightarrow{\sigma} (n',i{+}1)$ for $i\in [0..K)$, as well as the self-loop $(n,K) \xrightarrow{\sigma} (n,K)$ in the last copy.
	
	Then, in $\N'$ we replace every component $\N_q$ with $\N'_q$, while replacing a transition $n'\to n$ from a state $n'\in N'\setminus N'_q$ to a state $n\in N_q$ by the transition $n'\to (n,0)$, and adding for every transition $n\to n'$ from a state $n\in N_q$ to a state $n'\in N' \setminus N'_q$, the transitions $(n,i)\to n'$, for every $i\in[0..k]$.
	
	The only cycles in $\N'$ are self loops. Furthermore, $\N'$ represents a winning strategy of Adam since:
	i) Every run of $\N'$ that does not eventually remain in $\N'_q$ has a corresponding run in $\N$; and 
	ii) Every run of $\N'$ that eventually remains in $\N'_q$ is rejecting for Eve, as $q$ is a rejecting state of $\A$, and it corresponds to a word in the language of $\A$, by the claim.
\end{proof}


\section{Proofs of Section \ref{sec:No}}\label{app-sec:no}

\subsection{Automata with bounded resources}

\clBoundedStates*

\begin{proof}
	Consider the B\"uchi automaton $\A$ depicted in \cref{fig:SmallBuchi}, accepting words with a finite number of $a$'s.
     $\A$ is not history-deterministic, as Adam can win the letter game by choosing $b$ whenever Eve is in an odd state of $\A$, and $a$ whenever she is in an even state.
	
	We now show  that $\A$ can simulate every B\"uchi automaton $\A'\leq\A$ with up to $2n$ states. Observe that since $\A'$ cannot accept a word with infinitely many $a$'s, its maximally strongly connected components (MSCCs) cannot have both a transition on $a$ and an accepting state. We call an MSCC  accepting if it contains an accepting state, and rejecting otherwise.
	
	We define the following strategy for Eve in the simulation game between $\A$ and $\A'$: If Adam is in an accepting MSCC while she is in an odd state, she moves to the next even state. In other cases, she remains in her state.
	We claim that Eve's strategy is winning in $\Sim(\A,\A')$. Indeed, if Adam's run is accepting, he must eventually end in an accepting MSCC, and by Eve's strategy, so will she, unless she runs out of states. Since $\A'$ can have up to $2n$ MSCCs, Adam can move at most $n$ times from a rejecting MSCC (or from the initial state of $\A'$) to an accepting MSCC. (If he starts in an accepting MSCC, he still needs at least $2n{+}1$ MSCCs in order to have $n{+}1$ such moves to an accepting MSCC). Hence, Eve does not run out of states and wins $\Sim(\A,\A')$.	
\end{proof}

\Subject{Timed Automata with bounded clocks}\label{app-sec:bounded-clocks}
 
\clclock*

\begin{proof} For technical convenience, in this proof we use timestamps rather than delays to write timed words.

Consider the language of finite timed words $L_k \subseteq (\{a\}\times \mathbb{R}_{\geq 0})^{+} \cdot (\{b\} \times \mathbb{R}_{\geq 0}) \cdot (\{a\}\times \mathbb{R}_{\geq 0})^{+} $ over $\Sigma = \{a,b\}$ that contain a subword $w$ of the form $w = (a,t_1) (a,t_2) \cdots (a,t_k) (b,t) (a,t_1{+}1)$ $(a, t_2{+}1) \cdots$ $(a,t_k{+}1),$ such that $t_1 < t_2 < \cdots < t_k < t < t_1{+}1$. 

 We construct a timed automaton $\T$ with clocks $C \! = \! \{c_1,c_2,\cdots\! ,c_k\}$ that accepts $L_k$ and is guidable in $\mathbb{T}_k$, but not \HD. The automaton $\T$ consists of two components: i) a nondeterministic component in which all runs stay in till a $b$ is seen; and ii) a deterministic component in which all runs of the timed automaton go to and stay in after a $b$ is seen.

At a high level, the first component consists of states from where any set of clocks may be non-deterministically reset, upon reading an $a$. Upon reading $b$, the timed automaton moves to the second component. The second component checks deterministically whether any clock value is 1 when an $a$ is read, and it moves to an accepting sink when all clocks have tested positively for 1 at least once. 

We now describe the two components in detail. The first consists of states $(S,1)$, where $S$ is a subset of $C$, and $1$ indicates the first component. The set $S$ tracks which clocks are reset at least once in the run so far.  The state $(\emptyset,1)$ is the initial state in $\T$. 
The transitions on $a$ in the first component are of the form $((S,1),a,g_{\top},X,(S',1))$, where $g_{\top} = \top$ is the guard that accepts all clock valuations, $X$ is the set of clocks that are reset, and $S' = S \cup X$. 

The second component consists of a rejecting sink $\reject$ and states of the form $(S,2)$, where $S$ is a subset of $C$ and $2$ indicates the second component. The set $S$ tracks the clocks that have not positively tested for $1$ at an $a$-transition yet. Naturally, the state $(\emptyset,1)$ is the only accepting state in $\T$. The transitions on $a$ in the second component are:

\begin{enumerate}
    \item[A1.] $((S,2),a,g[y], \emptyset, (S',2))$, where $y\in [k]$ and $g[y] = (c_y = 1) \bigwedge_{z \in [k] \setminus \{y\}} (c_z\neq 1)$. The set $S'$ is then updated to be the set $S \setminus \{c_y\}$. 
    \item[A2.] $((S,2),a,g[2],\emptyset,\reject)$, where $g[2] = \bigvee_{1 \leq x < y \leq k} ((c_x =1) \land (c_y=1))$.
    \item[A3.] $((S,2),a,g[S \neq 1], \emptyset,(S,2))$, where $g[S \neq 1] = \bigwedge_{c \in S} (c \neq 1)$.
    \item[A4.] $(\reject,a,g[\top],\emptyset,\reject)$, where $g[\top] = \top$ is the guard that accepts all clock valuations.
\end{enumerate}
A transition in A1 can be taken if a clock value is $1$. This is the case if an $a$ is being repeated exactly one time-unit apart. Furthermore, the condition that requires other clocks to be not equal to $1$ makes sure that the $a$'s that are repeated occur at distinct times. If there are two clocks that have their values equal to $1$, then the automaton moves to $\reject$ via a transition in A2. The transitions in A3 are self loops, which are taken when none of the clock values are 1. Finally, there are additional self loops on the sink $\reject$ state. Since the guards $g[y]$ for each clock $y \in C$, $g[2]$ and $g[S \neq 1]$ are all disjoint, the second component is deterministic.

Additionally, there are the following $b$-transitions from the first to the second component:
\begin{enumerate}
    \item[B1] $((C,1),b, (g[b]),\emptyset,(C,2))$, where $S = C$ and $g[b] = \bigwedge_{c \in C} ((c$ $>0) \land (c<1))$ is the guard that checks that all clocks have valuation strictly between 0 and 1.
    \item[B2] $((C,1),b, \neg g[b],\emptyset,\reject)$
    \item[B3] $(S,1),b,g[\top], \emptyset, \reject)$, where $S \neq C$
\end{enumerate}
Only the transition B1 above does not immediately move the run to $\reject$, and can only be taken if the guard $g[b]$ is satisfied. Morally, the guard $g[b]$ ensures that $t_1+1>t$ and $t_k<t$ in the description of words belonging to $L_k$ given at the start of the proof.

We shall argue that:
\begin{enumerate}
    \item The timed automaton $\T$ accepts $L_k$.
    \item $\T$ is not \HD 
    \item $\T$ simulates every $\T'\leq \T$, such that $\T'$ has at most $k$ clocks.
\end{enumerate}

\Subject{$T$ accepts $L_k$} Consider a word $w$ in $L_k$, and let $t_1<t_2<\cdots<t_k$ be timestamps of $k$ distinct $a$'s that occur in $w$ which are repeated one time-unit later, with the timed letter $(b,t)$ occurring in between, such that $t_1 + 1 > t>t_k$. We describe an accepting run $\rho$ of $\T$ on the word $w$. The run $\rho$ resets the clock $c_i$ at time $t_i$, for each $i \in [k]$. Thus, after the letter $(a,t_k)$, $\rho$ is at the state $(C,1)$. At the letter $(b,t)$, the transition B1 is taken to the state $(C,2)$ since the guard $g[b]$ would be satisfied: The valuation of $c_i$ is $t-t_i$ for each $i \in [k]$, and $0<t-t_i<t-t_1<1$. The clock $c_i$ then tests positively for 1 at $(a,t_i+1)$ for each $i \in[k]$ (in the transition A1), and at the (first) occurrence of the timed letter $(a,t_k+1)$, $\rho$ moves to $(\emptyset,2)$, and stays there. As the state $(\emptyset,2)$ is accepting, the run $\rho$ is accepting.

Conversely, consider a word $w$ that is accepted by $\T$ via the accepting run $\rho$. Then, all the $k$ clocks must have been reset at least once in $\rho$ in the first component: let $(a,t_i)$ be the last timed letter at which the clock $c_i$ was reset. Without loss of generality, we may assume $t_1 \leq t_2 \leq \cdots \leq t_k$. When the timed letter $(b,t)$ is seen, $\rho$ must take the transition B1 to $(C,2)$ since it would be rejecting otherwise. Since the guard $g[b]$ must be satisfied to take the transition B1, we get that at the timed letter $(b,t)$, the valuation of each clock is strictly between 0 and 1. Thus, $t-t_1 < 1$ and $t-t_k>0$. In the second component, each clock $c_i$ in $C$ for $i \in [k]$ tests positively for $1$ on a timed letter $(a,t_i+1)$ during a transition in A1, since $\rho$ eventually reaches the state $(\emptyset,2)$. Furthermore, whenever the valuation of clock $c_i$ is 1 for the first time, we must have that the value of clock $c_j$ for $j \neq 1$ is not $1$. This implies that $t_1<t_2<\cdots<t_k$, and hence the word $w$ is in $L_k$.     

\Subject{$\T$ is not \HD} We describe a winning strategy for Adam in the letter game on $\T$. Adam starts with the timed letters $(a,\frac{1}{k+2}), (a,\frac{2}{k+2}),$ $\cdots, (a,\frac{k+1}{k+2}), (b,1)$ in the first $k+2$ rounds, regardless of Eve' transitions. After the $k+2$ rounds, Eve's token  can either be at $\reject$, or at the state $(C,2)$. In the former case, Adam gives the letters $(a,1+\frac{1}{k+2}),(a,1+\frac{2}{k+2}),\cdots,(a,1+\frac{k}{k+2})$ in the next $k$ rounds, and wins because the resulting word is in the language. Otherwise, Eve resets all clocks at least once in the run on her token so far. Consider the clock valuations after the transition corresponding to the timed letter $(b,1)$. Then, either  
\begin{enumerate}
    \item There are two clocks that have the same value, say $1-\frac{i}{k+2}$ due to them being last reset at the letter $(a,\frac{i}{k+2})$. If Adam gives letters $(a,1+\frac{1}{k+2}),(a,1+\frac{2}{k+2}),\cdots,(a,1+\frac{k+1}{k+2})$, the resulting word would be in the language, but Eve's token would eventually move to a rejecting sink state; or
    \item All clocks have distinct values. Then there is a timed letter $(a,\frac{j}{k+2})$ such that all clocks that were reset on $(a,\frac{j}{k+2})$ were reset later. In particular, it could also be the case that no clock was reset when the timed letter $(a,\frac{j}{k+2})$ was seen. Let $i \neq j$ be an element in $[k+1]$. Since all the clocks have distinct values, there is a clock $c$ with the value $a,1-\frac{i}{k+2}$ when the transition on the timed letter $(b,1)$ is taken. Suppose Adam gives the timed letters $(a,1+\frac{l}{k+2})$ for each $l \in [k+1] \setminus\{i\}$ in the next $k$ rounds of the letter game. Then the resulting word is accepting, while the valuation of clock $c$ is never 1 at any transition, causing the run in Eve's token to be rejecting.
\end{enumerate}
Thus, Adam has a winning strategy in the letter game, and hence $\T$ is not \HD.

\Subject{$\T$ is guidable in $\mathbb{T}_k$}
Let $\T'$ be a timed automaton in $\mathbb{T}_k$ accepting the language $L'$. We assume for convenience that $\T'$ has exactly $k$ clocks denoted by the set $C = \{c_1,c_2,\cdots,c_k\}$, same as that for $\T$. We give a strategy $\sigma$ for Eve in the simulation game $\Sim(\T,\T')$ between $\T$ and $\T'$ such that if Eve loses $\Sim(\T,\T')$ when playing according to $\sigma$, then $L(\T')$ is not a subset of $L(\T)$. This would imply that $\T$ is guidable in $\mathbb{T}_k$.

Suppose that after $i$ rounds of the simulation game, Eve's token is at the state $(S,1)$ with the valuation $\nu$, while Adam's token is at the state $q$ in $\T'$ with the valuation $\nu'$. In the $(i+1)^{th}$ round, if Adam selects the letter $b$ then Eve's transition in $\T$ is determined. Otherwise, suppose Adam selects the letter $(a,t)$ with the duration being $d$, and the transition $(q,a,g,X,q')$. Let $X'$ be the set of clocks that are integers in the valuation $\nu'+d$. Eve's strategy $\sigma$ takes the transition that resets the clocks $X \cup X'$ in $\T$. As the second component of $\T$ is determined, Eve has at most one transition available to her from states of the simulation game, which is the transition that $\sigma$ chooses. 

We will show that $\sigma$ is a winning strategy, unless $L(\T') \nsubseteq L(\T)$. For each valuation $\nu: C \xrightarrow{} \mathbb{R}$, let $\Frac(\nu) = \langle \Frac(\nu(c_1)), \Frac(\nu(c_2),$ $\cdots, \Frac(\nu(c_k)))\rangle$ be the $k$-tuple that consists of the fractional part of the clock valuations of $\nu$. An inductive argument gives us that the strategy $\sigma$ described above preserves the invariant that after each round of the simulation game, if Adam's token is at $q'$ with valuation $\nu'$ and Eve's token is at the state $(S,1)$ with valuation $\nu$ in $\T$, then $\Frac(\nu) = \Frac(\nu')$.

Suppose that there is a play of the simulation game that agrees with $\sigma$, but which Adam wins. Let $w$ be a timed word accepted by $\T'$ via the run $\rho'$ that Adam constructs in the simulation game; then Eve's run $\rho$ in $\T$ constructed according to $\sigma$ is rejecting. 

We assume that $w$ is of the form $w =(a,t_1) (a,t_2) \cdots (a,t_l) (b,t)$ $(a,t'_1) (a,t'_2) \cdots (a,t'_m)$, as otherwise $L(\T') \nsubseteq L(\T)$. Let $\{\alpha_1,\alpha_2,\cdots,$ $\alpha_n\} \in [l]$ be the largest set such that the timed letters $(a,t_{\alpha_i})$ and $(a,t_{\alpha_i}+1)$ both occur in the word $w$, $t-1< t_{\alpha_i} < t$ for each $i \in [n]$, and $t_{\alpha_i} \neq t_{\alpha_j}$ for all $1 \leq i < j \leq n$. We say that $n$ is the number of \emph{good pairs} in the word $w$. If $n<k$, then $w \notin L_k = L(\T)$, and hence $L(\T') \nsubseteq L(\T)$. Hence, suppose $n \geq k$. We will show that there is a word $\hat{w}$ that is accepted by $\T'$ but has strictly less than $k$ good pairs. This will imply that $L(\T') \nsubseteq L(\T)$.

Let $\nu'$ and $\nu$ be the valuations of the clocks in $\rho'$ and $\rho$ respectively, after the transition corresponding to the timed letter $(a,t_l)$ in the word $w$ is taken. Let Eve's token be at the state $(S,1)$ at this point, and $p$ be the maximum number of clocks in $S$ such that all of them have distinct valuations in $\nu$ of the form $t_l-t_{\alpha_i}$ for some $i$. Without loss of generality, suppose that the valuations of these $p$ clocks are $t_l -t_{\alpha_1},t_l-t_{\alpha_2}, \cdots,t_l-t_{\alpha_p}$, respectively. Note that the run $\rho$ on the suffix $(b,t) (a,t'_1) (a,t'_2) \cdots (a,t'_m)$ is uniquely determined in $\T$ from the state $(S,1)$ with valuation $\nu$ and, by assumption, is rejecting. Hence, $p$ is strictly less than $k$. Let $(a,t'_\beta)=(a,t_{\alpha_i}+1)$ for $\beta \in [m]$ be the first occurrence of the timed letter $(a,t_{\alpha_i}+1)$ in the word $w$, for an $i \in [n] \setminus [p]$. We claim that
\begin{claim}
Each of the clocks in the valuation $\nu''$ after the delay step of the transition on $(a,t'_{\beta})$ in $\rho'$ is non-integral. 
\end{claim}
\begin{proof}[Proof of the claim]
For each clock $c$, let $(a,t_{c})$ be the last timed letter at which the clock $c$ was reset in the run $\rho$. It follows that the clock $c$ was not reset in $\rho'$ after $(a,t_c)$ while the run $\rho$ was in the first component, i.e., at least until the timed letter $(b,t)$ is seen. If the clock $c$ was reset in $\rho'$ after $(b,t)$, then $\nu''(c)\leq t_{\alpha_i} + 1 -t < 1$. As $(a,t'_{\beta})$ is the first occurrence of a timed letter of the form $(\sigma,t_c)$, we know that $\nu''(c)>0$. Thus $\nu''(c)$ is not an integer. 

Hence, suppose that $(a,t_c)$ was the last timed letter at which the clock $c$ was reset in $\rho'$, before $(a,t'_{\beta})$. Then $\nu''(c) = t_{\alpha_i} + 1 -t_c$. If $\nu''(c)$ is an integer, then $t_{\alpha_i} - t_c$ is also an integer. As $\sigma$ resets a clock in $\rho$ in the first component whenever the valuation of the clock is $1$, we get that the clock $c$ must have been last reset at a timed letter of the form $(a,t_{\alpha_i})$ in the run $\rho$. Thus,  $\nu(c) = t_l - t_{\alpha_i}$, which is a contradiction to the fact that none of the clock valuations in $\nu$ equals $t_l - t_{\alpha_i}$. 
\end{proof}
Since each clock has a non-integral valuation in the delay step on the first $(a,t_{\alpha_i}+1)$, we can replace all occurrences of the timed letters of the form $(a,t_{\alpha_i}+1)$ in $w$ with $(a,t_{\alpha_i}+\epsilon)$ for some $\epsilon>0$ such that \begin{enumerate}
    \item $t_{\alpha_i}+\epsilon \neq t_{\alpha_j}$ for any $j \in [n] \setminus \{i\}$,
    \item $w'$ is still accepted in $\T'$ by an accepting run $\rho''$ that visits the same states $\rho'$.
\end{enumerate}
The number of good pairs in $w'$ now is $n-1$. We can iteratively repeat the above procedure on $w'$ to get a word $\hat{w}$ that has exactly $p$ good pairs and is accepted by $\T'$. As $p<k$, however, $\hat{w}$ is not in $L_k$, and hence, $L(\T')$ is not a subset of $L(\T)$.
\end{proof}



\end{document}